\documentclass[%
 reprint,
 raggedbottom,
 superscriptaddress,
 amsmath,amssymb,
 aps,
 pra,
]{revtex4-2}

\usepackage{graphicx}
\usepackage{dcolumn}
\usepackage{bm}
\usepackage{amsfonts}
\usepackage{amsthm}
\usepackage{mathrsfs}
\usepackage{bbold}
\usepackage{hyperref}

\makeatletter
\newcommand{\articlefigure}[3]{%
  \par\addvspace{8pt}%
  \noindent\begin{minipage}{\linewidth}%
    \centering
    \includegraphics[width=#1]{#2}%
    \def\@captype{figure}%
    \caption{#3}%
  \end{minipage}%
  \par\addvspace{8pt}%
}
\makeatother

\DeclareMathOperator{\Tr}{Tr}

\newtheorem{lemma}{Lemma}
\newtheorem{corollary}{Corollary}
\newtheorem{proposition}{Proposition}

\theoremstyle{definition}

\begin{document}

\title{Game of Nons: Non-Gaussian Nonunitary Environments with Nonlinear Coupling}

\author{A.\,Yu. Karasev}
\email{artem.karasev.01@gmail.com}
\affiliation{Department of Mathematical Methods for Quantum Technologies,
Steklov Mathematical Institute of Russian Academy of Sciences,
ul. Gubkina 8, Moscow 119991, Russia}

\author{A.\,E. Teretenkov}
\email{taemsu@mail.ru}
\affiliation{Department of Mathematical Methods for Quantum Technologies,
Steklov Mathematical Institute of Russian Academy of Sciences,
ul. Gubkina 8, Moscow 119991, Russia}

\date{\today}

\begin{abstract}
We develop a perturbative framework for systems coupled quadratically
to bosonic environments with quadratic Hamiltonians and linear Lindblad
operators. The method applies to both single-mode and multimode
environments whose initial states may be non-Gaussian and nonstationary.
By combining linear evolution in a doubled operator space with a
characteristic-function representation, we express ordered multitime
environmental correlation functions in terms of derivatives of the
initial characteristic function. These correlation functions enter
a recursive construction of the time-convolutionless generator at
arbitrary finite order. We illustrate the method with a qubit coupled
to a damped mode via a two-photon interaction. Initial Fock and thermal
states with identical first and second field moments produce different
second-order generators through their fourth moments.
\end{abstract}

\maketitle

	\section{Introduction}

	Perturbative calculations of reduced dynamics require ordered multitime environmental correlation functions and their contributions to the reduced generator. For nonlinear system--environment coupling, these functions involve higher field moments that are not determined by the field covariance. Here we develop a method for evaluating them for quadratic coupling to a bosonic environment with nonunitary free dynamics and a possibly non-Gaussian initial state. Related approaches include dissipaton equations of motion for quadratic bath coupling~\cite{XuLiuZhangYan2017,XuZhangZhangYan2018,ChenWangXuYan2023}, higher-order spectra of non-Gaussian quantum noise~\cite{WangClerk2020}, and nonlinear pseudomode constructions~\cite{ZhangLambertLiCirioLiang2025}. Generalized hierarchical equations also incorporate multitime bath statistics to describe non-Gaussian environments~\cite{HsiehCao2018}.

	The construction applies to multimode bosonic environments with
	quadratic generators of the free dynamics. For such environments,
	closed evolution equations for higher moments have been
	derived~\cite{NosalTeretenkov2022}. For nonunitary free dynamics, the Heisenberg evolution of multiplication superoperators mixes left and right multiplication. We account for this mixing in a doubled operator space. Related bosonic operator-space representations were developed in Refs.~\cite{ProsenSeligman2010,BarthelZhang2022}. Algebraic methods based on quadratic combinations of left and right multiplication superoperators have also been developed for multimode bosonic Lindblad dynamics~\cite{GaidashEtAl2025}.

	A related generating-function approach evaluates correlation-function moments for polynomial system--bath coupling within memory-kernel coupling theory~\cite{BiLiuDou2025}. There, moments obtained from repeated action of the full Liouvillian are used to reconstruct a two-time correlation function for a thermal harmonic reservoir. Here, we instead evaluate ordered multitime environmental correlations for insertion into the TCL expansion. 

	Our method constructs the reduced TCL generator at any prescribed perturbative order from the initial environmental characteristic function, including non-Gaussian and nonstationary states. We combine recursive TCL constructions~\cite{GasbarriFerialdi2018} and the moment-based formulation for composite open quantum systems~\cite{KarasevTeretenkov2023} with doubled operator-space representations~\cite{ProsenSeligman2010,BarthelZhang2022}. The construction retains the ordering of left and right multiplication channels in multitime correlations, beyond the propagation of equal-time field moments~\cite{NosalTeretenkov2022,FilippovManko2011}. The doubled-space evolution determines the coefficients of the linear
	and quadratic forms representing the environmental interaction channels,
	while derivatives of the initial characteristic function carry the
	dependence on the environmental state. Changing the initial environmental
	state therefore requires only replacing these derivatives, with the
	dynamical coefficients remaining unchanged.

	As an explicit application, we consider a qubit quadratically coupled to a single environmental mode. The qubit constitutes the finite-dimensional subsystem of interest \(A\), while restricting the environment to one mode keeps direct numerical checks of the composite dynamics tractable. By contrast, increasing the number of environmental modes rapidly increases the dimension of the truncated Hilbert space used in numerical experiments. This model illustrates how higher-order field moments contribute to
	the reduced dynamics already at the lowest nontrivial perturbative order. This complements nonlinear-coupling approaches based on Gaussian bath statistics~\cite{XuLiuZhangYan2017,ZhangLambertLiCirioLiang2025} by allowing the initial environmental state itself to be non-Gaussian and nonstationary.

	\section{\label{sec:reduced-tcl-expansion}From environmental correlations to reduced dynamics}

	\subsection{Composite system and initial state}

	We consider a composite open quantum system with Hilbert space \(\mathcal H=\mathcal H_A\otimes\mathcal H_B\), where subsystem \(A\) is the system of interest and subsystem \(B\) is its environment. Hats denote Hilbert-space operators, and calligraphic letters denote superoperators, with \(\mathcal H\) reserved for Hilbert spaces.

	The density operator \(\hat\rho_{AB}(t) \in \mathcal H\) satisfies
	\begin{equation}
		\frac{d}{dt}\hat\rho_{AB}(t)
		=
		\left(
		\mathcal L_0
		+
		\lambda\mathcal V
		\right)
		\hat\rho_{AB}(t),
		\label{eq:composite-master-equation}
	\end{equation}
	where \(\lambda\) is a small dimensionless parameter controlling the
	overall interaction strength.

	The free generator is assumed to be local,
	\begin{equation}
		\mathcal L_0
		=
		\mathcal L_0^A\otimes\mathcal I_B
		+
		\mathcal I_A\otimes\mathcal L_0^B,
	\end{equation}
	where \(\mathcal I_A\) and \(\mathcal I_B\) are the identity
	superoperators on the corresponding subsystems.
	We assume that \(\mathcal L_0^A\) and \(\mathcal L_0^B\) are time-independent
	generators of completely positive trace-preserving semigroups.  No restriction
	on the dimension of subsystem \(A\) is imposed at this stage. The assumptions on the quadratic environmental
	generator and the
	field moments required at each finite order are specified in
	Sec.~\ref{sec:quadratic-environment}.
	We denote the free environmental semigroup by
	\(\mathcal T_t=e^{t\mathcal L_0^B}\).

	The interaction superoperator is given by the commutator with
	the interaction Hamiltonian,
	\begin{equation}
		\mathcal V
		=
		-i[\hat H_I,\bullet \;].
	\end{equation}
	We assume that it admits a finite decomposition
	\begin{equation}
		\mathcal V
		=
		\sum_{i=1}^{s}
		\mathcal V_i^A
		\otimes
		\mathcal V_i^B.
	\end{equation}
	The factors \(\mathcal V_i^A\) and \(\mathcal V_i^B\) are superoperators on
	the corresponding subsystems and need not separately be dynamical generators.

	To derive the reduced evolution equation, we employ the projection-operator approach with the standard Argyres--Kelley projector~\cite{ArgyresKelley1964}. For every trace-class operator \(\hat X\) on \(\mathcal{H}\), it is defined by
\begin{equation}
	\mathcal P\hat X
	=
	\operatorname{tr}_B[\hat X]
	\otimes
	\hat\rho_B,
	\qquad
	\mathcal P^2=\mathcal P.
\end{equation}
We assume a factorized initial state with the same bath state
\(\hat\rho_B\) as in the definition of the projector,
\begin{equation}
	\hat\rho_{AB}(0)=\hat\rho_A(0)\otimes\hat\rho_B.
\end{equation}
Thus, the initial state is consistent with the projector,
so that \(\mathcal P\hat\rho_{AB}(0)=\hat\rho_{AB}(0)\) and
\((\mathcal I_{AB}-\mathcal P)\hat\rho_{AB}(0)=0\).
	
	In this work, we restrict the free dynamics of subsystem \(A\) to
	unitary evolution, so its inverse propagator is well defined.
	For the dissipative environment \(B\), the inverse free propagator
	is used formally as an algebraic tool. With this convention,
	the interaction-picture propagator is
	\begin{equation}
		\mathcal U_I(t)
		=
		e^{-t\mathcal L_0}
		e^{t(\mathcal L_0+\lambda\mathcal V)},
	\end{equation}
	where the factor \(e^{-t\mathcal L_0}\) removes the free evolution
	generated by \(\mathcal L_0\). The resulting propagator satisfies
	\begin{equation}
		\frac{d}{dt}\mathcal U_I(t)
		=
		\lambda\mathcal V(t)\mathcal U_I(t),
		\qquad
		\mathcal U_I(0)=\mathcal I_{AB},
	\end{equation}
	with the interaction-picture generator
	\begin{equation}
		\mathcal V(t)
		=
		e^{-t\mathcal L_0}
		\mathcal V
		e^{t\mathcal L_0}.
	\end{equation}
Applying the projector \(\mathcal P\) to this evolution defines
the reduced interaction-picture propagator \(\mathcal R(t)\).
Its action on a system operator \(\hat X_A\) is obtained by evolving
\(\hat X_A\otimes\hat\rho_B\) and tracing out the bath,
	\begin{equation}
		\mathcal R(t)\hat X_A
		=
		\operatorname{tr}_B\!\left[
		\mathcal U_I(t)
		\left(\hat X_A\otimes\hat\rho_B\right)
		\right].
	\end{equation}
	Equivalently, the projected evolution can be written as
	\begin{equation}
		\mathcal P\mathcal U_I(t)\mathcal P
		\left(\hat X_A\otimes\hat\rho_B\right)
		=
		\left(\mathcal R(t)\hat X_A\right)\otimes\hat\rho_B.
		\label{eq:reduced-interaction-propagator}
	\end{equation}

	\subsection{\label{sec:moment-functions}Recursive construction of the TCL generator}
		Our object of interest is the time-local generator \(\mathcal K(t)\)
	of the reduced evolution, which governs the interaction-picture state
	according to
	\begin{equation}
		\frac{d}{dt}\hat\rho_A^I(t)
		=\mathcal K(t)\hat\rho_A^I(t),
		\qquad
		\hat\rho_A^I(t)=\mathcal R(t)\hat\rho_A(0).
		\label{eq:reduced-state-tcl-equation}
	\end{equation}
We expand \(\mathcal K(t)\) formally in powers of the small dimentionless
parameter \(\lambda\),
	\begin{equation}
		\mathcal K(t)
		=
		\sum_{n=1}^{\infty}\lambda^n\mathcal K_n(t),
	\end{equation}
	expressing its coefficients in time-convolutionless (TCL) form~\cite{ShibataTakahashiHashitsume1977,ChaturvediShibata1979}. The coefficients \(\mathcal K_n(t)\) are the ordered cumulant contributions of the Kubo--van Kampen expansion~\cite{Kubo1962,VanKampen1974I,VanKampen1974II}. The expansion is understood purely formally: we neither establish its convergence nor derive bounds on the remainders after truncation. These questions are left for future research.

	The moments \(\mathcal M_n(t)\) are defined by the formal Dyson expansion
	of the reduced propagator
	\begin{equation}
		\mathcal R(t)
		=
		\mathcal I_A
		+
		\sum_{n=1}^{\infty}\lambda^n\mathcal M_n(t).
	\end{equation}
	Substituting this expansion and that of \(\mathcal K(t)\) into
	\(\dot{\mathcal R}(t)=\mathcal K(t)\mathcal R(t)\) and comparing
	equal powers of \(\lambda\) gives the recurrence
	\begin{equation}
		\label{eq:cumulant-recurrence}
		\mathcal K_N(t)
		=
		\dot{\mathcal M}_N(t)
		-
		\sum_{r=1}^{N-1}
		\mathcal K_r(t)\mathcal M_{N-r}(t).
	\end{equation}
	The order of the superoperator factors in Eq.~\eqref{eq:cumulant-recurrence} is essential. This recurrence is the noncommutative analogue of passing from moments to cumulants and is closely related, up to notation and ordering conventions, to the recursive construction of the non-Markovian time-convolutionless generator developed in Ref.~\cite{GasbarriFerialdi2018}.
	
	The method yields systematic corrections at arbitrary orders
	of perturbation theory. We give the complete expressions needed
	through second order, where two-time environmental correlations
	first enter beyond the single-time averages of the first-order term.
	The fourth-order contribution used in the numerical examples below
	is likewise obtained algebraically using the general algorithm
	implemented in Wolfram Mathematica.

	Practical implementations of fourth-order TCL generators have been used to identify corrections missed by second-order Bloch--Redfield dynamics~\cite{CrowderLampertManchandaEtAl2024} and to benchmark completely positive approximations to the Redfield equation~\cite{Davidovic2022}. Comparisons with numerically exact spin--boson dynamics identify parameter regimes in which TCL4 improves on TCL2~\cite{ChenCrowderXiangDavidovic2025}. A complementary geometric--arithmetic construction enforces complete positivity while retaining nonsecular terms of the Redfield equation~\cite{Davidovic2020}. These works primarily address stationary reservoirs and linear system--environment coupling.

	For composite open quantum systems with nonunitary free dynamics, expressions for the moments \(\mathcal M_N(t)\) in terms of iterated time integrals and projection superoperators are given in Ref.~\cite{KarasevTeretenkov2023}. For our purposes, it is convenient to evaluate the action of the projectors before performing the time integrations. This expresses the moments through ordered multitime environmental correlation functions and system superoperators.
	Trace preservation allows the outermost factor
	\(\mathcal T_{-t_N}\) in the interaction-picture expression to be
	omitted under the trace. The resulting forward-propagator representation
	defines the correlation without applying inverse environmental evolution
	to a state. For \(0\leq t_1\leq\cdots\leq t_N\), it reads
	\begin{align}
		&m_{i_N\dots i_1}(t_N,\dots,t_1)
		\nonumber\\
		&=
		\operatorname{tr}_B
		\Bigl[
		\mathcal V_{i_N}^B
		\mathcal T_{t_N-t_{N-1}}
		\mathcal V_{i_{N-1}}^B
		\cdots
		\mathcal T_{t_2-t_1}
		\mathcal V_{i_1}^B
		\mathcal T_{t_1}
		\hat\rho_B
		\Bigr].
		\label{eq:environmental-multitime-moment-forward}
	\end{align}
	Note that \(\hat\rho_B\) is generally nonstationary, so this expression
	depends not only on time differences.

	In terms of these correlation functions, the \(N\)-th moment is
	\begin{align}
		\mathcal M_N(t)
		&=
		\sum_{i_N,\dots,i_1}
		\int_{0<t_1<\cdots<t_N<t}
		dt_1\cdots dt_N
		\nonumber\\
		&\quad\times
		\mathcal V_{i_N}^A(t_N)
		\cdots
		\mathcal V_{i_1}^A(t_1)
		m_{i_N\dots i_1}(t_N,\dots,t_1),
		\label{eq:MN-environmental-moments}
	\end{align}
	where 	\(\mathcal V_i^A(t)=e^{-t\mathcal L_0^A}\mathcal V_i^A e^{t\mathcal L_0^A}\).
	
	The derivative \(\dot{\mathcal M}_N(t)\) is obtained by removing
	the outermost time integral and setting \(t_N=t\) in the integrand,
	while retaining the ordering \(0<t_1<\cdots<t_{N-1}<t\).
	At first order, the recurrence gives
	\begin{align}
		\mathcal K_1(t)
		=
		\dot{\mathcal M}_1(t)
		=
		\sum_i
		\mathcal V_i^A(t)m_i(t),
	\end{align}
	where \(m_i(t)\) is the \(N=1\) case of
	Eq.~\eqref{eq:environmental-multitime-moment-forward}.
	At second order,
	Eq.~\eqref{eq:cumulant-recurrence} then gives
\begin{align}
	\mathcal K_2(t)
	&=
	\dot{\mathcal M}_2(t)
	-
	\mathcal K_1(t)\mathcal M_1(t)
	\nonumber\\
	&=
	\sum_{i,j}
	\mathcal V_i^A(t)
	\int_0^t dt_1\,
	\mathcal V_j^A(t_1)
	\nonumber\\
	&\qquad{}\times
	\left[
	m_{ij}(t,t_1)
	-
	m_i(t)m_j(t_1)
	\right].
\end{align}
	Thus, the second-order contribution is determined by the ordered
	"superoperator covariance". For quadratic coupling, this covariance
	involves field moments up to fourth order. Obviously, perturbative order \(N\) in our quadratic setup requires field moments up to order \(2N\).

	The preceding construction leads to three main steps for calculating
	the reduced generator at any fixed perturbative order:
	\begin{enumerate}
		\item Compute the Heisenberg evolution
		\(\mathcal V_i^B(t)\) of every environmental superoperator entering
		the interaction.
		\item Evaluate the ordered multitime environmental correlation
		functions \(m_{i_r\dots i_1}\) required up to the chosen order,
		using the characteristic function of \(\hat\rho_B\).
		\item Assemble \(\mathcal M_r(t)\) and \(\dot{\mathcal M}_r(t)\)
		using the ordered time integral in
		Eq.~\eqref{eq:MN-environmental-moments} and the differentiation
		rule stated immediately below it, then construct
		\(\mathcal K_r(t)\) successively from
		Eq.~\eqref{eq:cumulant-recurrence}.
	\end{enumerate}
	
	Section~\ref{sec:quadratic-environment} develops the first two steps
	in detail for a multimode bosonic environment.

	\section{\label{sec:quadratic-environment}Tracking correlations in a multimode environment}

	To evaluate the environmental correlation functions entering the TCL
	construction for both single-mode and multimode environments, it is
	convenient to use the operator and moment formalism developed in
	Refs.~\cite{Teretenkov2019,NosalTeretenkov2022,MeretukovTeretenkov2026}.
	This formalism also extends to fermionic environments.
	We use arrows for vectors \(\vec{x}\) and uppercase letters for matrices  \(X\).
	For an environment of \(d\) bosonic modes, the vector of creation and
	annihilation operators is denoted by
	\begin{equation}
		\mathfrak a
		=
		(\hat a_1,\dots,\hat a_d,
		\hat a_1^\dagger,\dots,\hat a_d^\dagger)^{\mathsf T}.
	\end{equation}
	We also introduce the symplectic and exchange matrices
	\begin{equation}
		J
		=
		\begin{pmatrix}
			0&-I_d\\
			I_d&0
		\end{pmatrix},
		\qquad
		E
		=
		\begin{pmatrix}
			0&I_d\\
			I_d&0
		\end{pmatrix},
		\label{eq:exchange-matrix}
	\end{equation}
	and the involution
	\begin{equation}
		\widetilde{\vec{x}}=E\vec{x}^{\,*},
		\qquad
		\widetilde X=EX^* E,
		\label{eq:tilde-conjugation}
	\end{equation}
	for vectors and matrices, respectively.  The star denotes elementwise
	complex conjugation.  The canonical commutation relations take the compact
	form
	\begin{equation}
		[\mathfrak a,\mathfrak a^{\mathsf T}]
		=
		-J.
	\end{equation}

	The free environmental dynamics is generated by the time-independent
	quadratic GKSL generator
	\begin{align}
		\mathcal L_0^B
		=
		&-i
		\left[
		\frac12
		\mathfrak a^{\mathsf T}H\mathfrak a,
		\,\bullet
		\right]
		\nonumber\\
		&+
		\mathfrak a^{\mathsf T}\bullet\,\Gamma\,\mathfrak a
		-
		\frac12
		\mathfrak a^{\mathsf T}\Gamma^{\mathsf T}\mathfrak a\,\bullet
		-
		\frac12
		\bullet\,
		\mathfrak a^{\mathsf T}\Gamma^{\mathsf T}\mathfrak a .
	\end{align}
	In this work, we assume that the linear Hamiltonian term is absent,
	which simplifies the resulting expressions. The coefficient matrices satisfy
	\begin{equation}
		H=H^{\mathsf T}=\widetilde H,
		\qquad
		\Gamma^{\mathsf T}=\widetilde\Gamma,
		\qquad
		\Gamma E\succeq0.
		\label{eq:quadratic-gksl-conditions}
	\end{equation}
	These conditions ensure that the resulting generator is of GKSL form.

    Each environmental interaction factor \(\mathcal V_i^B\) is assumed
	to be a polynomial of degree at most two in left and right multiplication by
	the components of \(\mathfrak a\).

	The key to applying the same algorithm to different environmental
	states is to represent the density operator through its characteristic
	function. This representation separates the dependence on the initial
	state from the evolution of the interaction channels.
	For the reference state \(\hat\rho_B\), the characteristic function
	is defined by
\begin{equation}
	\chi(\vec{z})
	=
	\operatorname{tr}_B
	\Bigl[
	e^{i\vec{z}^{\mathsf T}\mathfrak a}
	\hat\rho_B
	\Bigr],
\end{equation}
where
\begin{equation}
	\vec{z}
	=
	(z_1,\dots,z_d,
	z_1^*,\dots,z_d^*)^{\mathsf T}.
\end{equation}
	Thus the dependence on a particular environmental state is entirely encoded in
	\(\chi(\vec{z})\). At perturbative order \(N\), we require the characteristic function
	\(\chi\) to be continuously differentiable up to order \(2N\)
	in a neighborhood of the origin, since each environmental interaction
	factor is quadratic in the creation and annihilation operators.
	
	Using the inverse Weyl transform~\eqref{eq:weyl-inversion}, we rewrite
	the environmental correlation functions defined
	in Eq.~\eqref{eq:environmental-multitime-moment-forward} as
	\begin{align}
		&m_{i_N\dots i_1}(t_N,\dots,t_1)
		\nonumber\\
		&=
		\frac1{\pi^d}
		\int
		\mathrm d^{2d}\vec{z}\,
		\chi(\vec{z})
		\operatorname{tr}_B
		\Bigl[
		\mathcal V_{i_N}^B
		\mathcal T_{t_N-t_{N-1}}
		\mathcal V_{i_{N-1}}^B
		\cdots
		\nonumber\\
		&\qquad\qquad\times
		\mathcal T_{t_2-t_1}
		\mathcal V_{i_1}^B
		\mathcal T_{t_1}
		e^{-i\vec{z}^{\mathsf T}\mathfrak a}
		\Bigr].
	\end{align}

	In what follows, we denote the exponential
	\(e^{-i\vec{z}^{\mathsf T}\mathfrak a}\) by \(\hat W(\vec{z})\).
	The evaluation of perturbative corrections therefore reduces to
	calculating the Heisenberg evolution of multiplication superoperators
	acting on these Weyl operators.

	\subsection{Linear evolution in a doubled space}

	Evaluating the environmental correlation functions requires both left and right
	multiplications by environmental operators.
	To treat these actions in a unified manner, we introduce the doubled vector of
	multiplication superoperators
	\begin{equation}
		\mathfrak A
		=
		\begin{pmatrix}
			\mathfrak a\bullet \\
			\bullet E\mathfrak a
		\end{pmatrix},
	\label{eq:doubled-vector}
	\end{equation}
	where \(E\) is the exchange matrix defined in Eq.~\eqref{eq:exchange-matrix}. The first block corresponds to left multiplication, whereas the second block corresponds to right multiplication. Bold symbols denote vectors
	\(\vec{\mathbf{x}}\) and matrices \(\mathbf X\) in the doubled space.

	The doubled vector satisfies canonical commutation relations
	\begin{equation}
		[\mathfrak A,\mathfrak A^{\mathsf T}]
		=
		-\mathbf{J},
		\qquad
		\mathbf{J}
		=
		J\oplus J.
	\end{equation}

	For the time-independent quadratic GKSL generator considered here, the
	Liouvillian can be written as
	\begin{equation}
		\mathcal L_0^B
		=
		\frac12
		\mathfrak A^{\mathsf T}
		\mathbf{L}
		\mathfrak A
		+
		\frac12\operatorname{Tr}[\Gamma J]\mathcal I_B,
		\label{eq:doubled-liouvillian}
	\end{equation}
	where the symmetric matrix \(\mathbf{L}\in\mathbb C^{4d\times4d}\) is given in
	Appendix~\ref{app:doubled-heisenberg}, Eq.~\eqref{eq:doubled-liouvillian-matrix}.

	We define the Heisenberg evolution of the multiplication superoperators by
	\begin{equation}
		\mathfrak A(t)
		=
		\mathcal T_{-t}
		\mathfrak A
		\mathcal T_t.
	\label{eq:doubled-evolution-definition}
	\end{equation}
	This is a similarity transformation on superoperators, rather than
	the adjoint evolution of Hilbert-space observables. It preserves
	products of superoperators, whereas dissipative evolution of observables
	generally does not preserve their products. The key advantage of these doubled-space operators is that their
	evolution is governed by a closed linear system. Since the generator
	is time independent, the propagator is a matrix exponential:
	\begin{equation}
		\mathfrak A(t)
		=
		\mathbf{G}(t)\mathfrak A,
		\qquad
		\mathbf{G}(t)
		=
		e^{-\mathbf{J}\mathbf{L}t}.
		\label{eq:doubled-linear-propagator}
	\end{equation}

	The commutation relations and the evolution formula are proved in
	Appendix~\ref{app:doubled-heisenberg}.

	To relate this evolution to the forward-propagator correlations,
	write \(\mathcal V_i^B(t)=\mathcal T_{-t}\mathcal V_i^B\mathcal T_t\).
	On the domain where this similarity transformation is defined,
	the semigroup composition law gives
	\begin{align}
		&\mathcal V_{i_N}^B(t_N)\cdots\mathcal V_{i_1}^B(t_1)
		\nonumber\\
		&=\mathcal T_{-t_N}\mathcal V_{i_N}^B
		\mathcal T_{t_N-t_{N-1}}\mathcal V_{i_{N-1}}^B
	\cdots\mathcal T_{t_2-t_1}\mathcal V_{i_1}^B\mathcal T_{t_1}.
	\end{align}
	Taking the trace removes only the leftmost \(\mathcal T_{-t_N}\),
	so Eq.~\eqref{eq:environmental-multitime-moment-forward} is equivalently
	\begin{align}
		&m_{i_N\dots i_1}(t_N,\dots,t_1)
		\nonumber\\
		&=\operatorname{tr}_B\!\left[
		\mathcal V_{i_N}^B(t_N)\cdots\mathcal V_{i_1}^B(t_1)\hat\rho_B
		\right].
		\label{eq:environmental-multitime-moment-interaction}
	\end{align}
	The forward representation defines the correlation through the
	environmental dynamics, while the product representation evaluates
	it algebraically by evolving the coefficients of the interaction
	channels with \(\mathbf G(t)\) and taking the expectation value
	in the initial environmental state.

	Since similarity transformations preserve products, each environmental
	interaction factor evolves as
	\begin{equation}
		\mathcal V_i^B(t)
		=
		c_i\mathcal I_B
		+
		\vec{\mathbf{v}}_i(t)^{\mathsf T}\mathfrak A
		+
		\mathfrak A^{\mathsf T}\mathbf{V}_i(t)\mathfrak A,
	\end{equation}
	with
	\begin{equation}
		\vec{\mathbf{v}}_i(t)
		=
		\mathbf{G}(t)^{\mathsf T}\vec{\mathbf{v}}_i,
		\qquad
		\mathbf{V}_i(t)
		=
		\mathbf{G}(t)^{\mathsf T}\mathbf{V}_i\mathbf{G}(t).
	\end{equation}
	Here \(c_i\in\mathbb C\), \(\vec{\mathbf{v}}_i=\vec{\mathbf{v}}_i(0)\in\mathbb C^{4d}\), and
	\(\mathbf{V}_i=\mathbf{V}_i(0)\in\mathbb C^{4d\times4d}\), with \(\mathbf{V}_i=\mathbf{V}_i^{\mathsf T}\),
	specify the initial interaction factor.
	Consequently, each environmental correlation function is reduced to the trace of an ordered
	product of time-dependent polynomials in the static doubled vector
	\(\mathfrak A\). The next subsection converts these products into differential
	operators in the original phase space.

	\subsection{Extracting moments from the characteristic function}
	
	After substituting the inverse Weyl transform, left and right
	multiplication by creation and annihilation operators can be represented
	by differential operators acting on \(\hat W(\vec z)\).
	Moving these differential operators outside the trace leaves
	\(\operatorname{tr}_B\hat W(\vec z)\), which is proportional to
	the delta function at the origin. Integration by parts then expresses
	the environmental moments in terms of derivatives of the characteristic
	function, eliminating the Hilbert-space operators from the calculation.
	
	Differentiation with respect to \(z_k\) and \(z_k^*\) means
	Wirtinger differentiation. We use the vector
	\begin{equation}
		\partial_{\vec{z}}
		=
		\left(
		\partial_{z_1},\ldots,\partial_{z_d},
		\partial_{z_1^*},\ldots,\partial_{z_d^*}
		\right)^{\mathsf T},
	\end{equation}

Once the time evolution of the linear or quadratic forms in the doubled
space has been determined, the remaining task is to evaluate their action
on the Weyl kernel in the original phase-space variables. To do so, we
separate the coefficients associated with left and right multiplication:
\begin{equation}
	\vec{\mathbf{x}}=
	\begin{pmatrix}
		\vec{x}_L\\
		\vec{x}_R
	\end{pmatrix},
	\qquad
	\vec{x}_L,\vec{x}_R\in\mathbb C^{2d}.
\end{equation}
The two blocks of \(\mathfrak A\) represent left multiplication by
\(\mathfrak a\) and right multiplication by \(E\mathfrak a\).
When acting on \(\hat W(\vec z)\), left and right multiplication
by the same field operator give identical derivative terms but
coordinate-dependent terms of opposite sign. Since the right block contains \(E\mathfrak a\), its
coefficient vector enters these combinations as \(E\vec{x}_R\).
This motivates the definitions
\begin{equation}
	\vec{p}_{\vec{\mathbf{x}}}=\vec{x}_L+E\vec{x}_R,
	\qquad
	\vec{q}_{\vec{\mathbf{x}}}=\vec{x}_L-E\vec{x}_R.
\end{equation}
For subsequent calculations, we write these combinations as
\(\vec{p}_{\vec{\mathbf{x}}}=P\vec{\mathbf{x}}\) and
\(\vec{q}_{\vec{\mathbf{x}}}=Q\vec{\mathbf{x}}\), where
\begin{equation}
	P=\begin{pmatrix}I_{2d}&E\end{pmatrix},
	\qquad
	Q=\begin{pmatrix}I_{2d}&-E\end{pmatrix}.
\end{equation}
The action of a linear form in the doubled space can then be
written as
\begin{equation}
	(\vec{\mathbf{x}}^{\mathsf T}\mathfrak A)\hat W(\vec z)
	=\mathscr D_{\vec{\mathbf{x}}}\hat W(\vec z),
\end{equation}
where
\begin{equation}
	\mathscr D_{\vec{\mathbf{x}}}
	=i\left(
	\vec{p}_{\vec{\mathbf{x}}}^{\mathsf T}\partial_{\vec z}
	+\frac12\vec{q}_{\vec{\mathbf{x}}}^{\mathsf T}J\vec z
	\right).
	\label{eq:def-D-script-x}
\end{equation}
	The trace of a linear form is then
	\begin{equation}
		\operatorname{tr}_B
		\Bigl[
		\vec{\mathbf{x}}^{\mathsf T}\mathfrak A \hat W(\vec{z})
		\Bigr]
		=
		\pi^d\mathscr D_{\vec{\mathbf{x}}}\delta^{(2d)}(\vec{z}).
	\end{equation}
	The ordered-string version of the preceding identity is
	\begin{equation}
		\label{eq:ordered-doubled-string}
		\operatorname{tr}_B
		\Bigl[
		(\vec{\mathbf{x}}_1^{\mathsf T}\mathfrak A)
		\cdots
		(\vec{\mathbf{x}}_m^{\mathsf T}\mathfrak A)
		\hat W(\vec{z})
		\Bigr]
		=
		\pi^d
		\mathscr D_{\vec{\mathbf{x}}_m}
		\cdots
		\mathscr D_{\vec{\mathbf{x}}_1}
		\delta^{(2d)}(\vec{z}).
	\end{equation}
	Thus the differential composition is written in reverse order relative to
	the original multiplication-superoperator string. Define the ordered
	environmental moments by
	\begin{equation}
		m(\vec{\mathbf{x}}_1,\ldots,\vec{\mathbf{x}}_m)
		=
		\operatorname{tr}_B
		\Bigl[
		(\vec{\mathbf{x}}_1^{\mathsf T}\mathfrak A)\cdots
		(\vec{\mathbf{x}}_m^{\mathsf T}\mathfrak A)\hat\rho_B
		\Bigr].
	\end{equation}
	The inverse characteristic-function representation gives
	\begin{equation}
		\label{eq:environmental-moment-differential}
		m(\vec{\mathbf{x}}_1,\ldots,\vec{\mathbf{x}}_m)
		=
		\int \mathrm d^{2d}\vec{z}\,
		\chi(\vec{z})
		\mathscr D_{\vec{\mathbf{x}}_m}\cdots\mathscr D_{\vec{\mathbf{x}}_1}
		\delta^{(2d)}(\vec{z}).
	\end{equation}

	Moments of products of quadratic forms follow directly from the
	preceding formula by expanding each quadratic form into products
	of linear factors. For any rank-one decomposition
	\begin{equation}
		\label{eq:rank-one-decomposition}
		\mathbf{X}
		=
		\sum_\alpha
		\vec{\mathbf{x}}_\alpha \vec{\mathbf{y}}_\alpha^{\mathsf T},
	\end{equation}
	the quadratic form satisfies
	\begin{equation}
		\operatorname{tr}_B
		\Bigl[
		(\mathfrak A^{\mathsf T}\mathbf{X}\mathfrak A)
		\hat W(\vec{z})
		\Bigr]
		=
		\pi^d
		\sum_\alpha
		\mathscr D_{\vec{\mathbf{y}}_\alpha}
		\mathscr D_{\vec{\mathbf{x}}_\alpha}
		\delta^{(2d)}(\vec{z}).
	\end{equation}
	For decompositions
	\begin{equation}
		\mathbf{X}_\ell
		=
		\sum_{\alpha_\ell}
		\vec{\mathbf{x}}_{\ell\alpha_\ell}
		\vec{\mathbf{y}}_{\ell\alpha_\ell}^{\mathsf T},
		\qquad
		\ell=1,\ldots,N,
	\end{equation}
	the corresponding product is therefore reduced according to
	\begin{align}
		&
		\operatorname{tr}_B
		\Bigl[
		(\mathfrak A^{\mathsf T}\mathbf{X}_1\mathfrak A)
		\cdots
		(\mathfrak A^{\mathsf T}\mathbf{X}_N\mathfrak A)
		\hat\rho_B
		\Bigr]
		\nonumber\\
		&=
		\sum_{\alpha_1,\ldots,\alpha_N}
		m
		\left(
		\vec{\mathbf{x}}_{1\alpha_1},\vec{\mathbf{y}}_{1\alpha_1},
		\ldots,
		\vec{\mathbf{x}}_{N\alpha_N},\vec{\mathbf{y}}_{N\alpha_N}
		\right).
	\end{align}
	The differential identities and the integration-by-parts rule are derived
	in Appendix~\ref{app:differential-representation}.
	Explicit reduction formulas are collected in
	Appendix~\ref{app:wick}.

\subsection{\label{sec:normal-ordering}Algebraic formulas for environmental moments}

The differential representation obtained above allows us to express
the environmental moments of ordered operator strings as finite
sums of derivatives of the characteristic function evaluated at
the origin. The coefficients account for the order of the original
operators. This gives the following formula, valid for arbitrary
environmental states with the required moments.

\begin{proposition}[Algebraic reduction of environmental moments]
		\label{prop:environmental-moments}
		If the reference state possesses moments up to order \(m\), then
		\begin{align}
			&m(\vec{\mathbf{x}}_1,\ldots,\vec{\mathbf{x}}_m)
			\nonumber\\
			&=
			\sum_{\pi\in\operatorname{Pair}_{\leq}(m)}
			\left(
			\prod_{\substack{\{r,s\}\in\pi\\r<s}}
			\frac12
			\vec{q}_{\vec{\mathbf{x}}_s}^{\mathsf T}J\vec{p}_{\vec{\mathbf{x}}_r}
			\right)
			(-i)^{m-2|\pi|}
			\nonumber\\
			&\quad\times
			\left.
			\left(
			\prod_{k\notin\operatorname{supp}\pi}
			\vec{p}_{\vec{\mathbf{x}}_k}^{\mathsf T}\partial_{\vec{z}}
			\right)
			\chi(\vec{z})
			\right|_{\vec{z}=\vec{0}}.
		\end{align}
		Here \(\operatorname{Pair}_{\leq}(m)\) denotes the set of all partial
		pairings of \(\{1,\ldots,m\}\), \(\operatorname{supp}\pi\) is the set
		of paired indices, and \(|\pi|\) is the number of pairs. The orientation
		of each contraction is fixed by the original operator order \(r<s\).
	\end{proposition}

	
	For a Gaussian reference state,
		\begin{equation}
			\chi(\vec{z})
			=
			\exp\left(
			-\frac12\vec{z}^{\mathsf T}\Sigma\vec{z}
			+
			i\vec{z}^{\mathsf T}\vec{\mu}
			\right).
			\label{eq:gaussian-characteristic-function}
		\end{equation}
	Thus all higher field moments are determined by \(\vec{\mu}\) and \(\Sigma\),
	with no independent higher-order state data. Combining this property with
	the ordering contractions gives the following Wick formula.

	\begin{corollary}[Gaussian environmental moments]
		\label{cor:gaussian-environmental-moments}
		Let the reference state be Gaussian, with mean and covariance defined
		above.
		For \(r<s\), define
		\begin{align}
			\mu(\vec{\mathbf{x}}_r)
			&=
			\vec{p}_{\vec{\mathbf{x}}_r}^{\mathsf T}\vec{\mu},
			\\
			\kappa(\vec{\mathbf{x}}_r,\vec{\mathbf{x}}_s)
			&=
			\vec{p}_{\vec{\mathbf{x}}_s}^{\mathsf T}\Sigma \vec{p}_{\vec{\mathbf{x}}_r}
			+
			\frac12\vec{q}_{\vec{\mathbf{x}}_s}^{\mathsf T}J\vec{p}_{\vec{\mathbf{x}}_r}.
		\end{align}
		Then
		\begin{equation}
			\begin{aligned}
			&m(\vec{\mathbf{x}}_1,\ldots,\vec{\mathbf{x}}_m)
			\\
			&=
			\sum_{\pi\in\operatorname{Part}_{1,2}(m)}
			\left(
			\prod_{\{r\}\in\pi}\mu(\vec{\mathbf{x}}_r)
			\right)
			\left(
			\prod_{\substack{\{r,s\}\in\pi\\r<s}}
			\kappa(\vec{\mathbf{x}}_r,\vec{\mathbf{x}}_s)
			\right),
		\end{aligned}
		\label{eq:environmental-wick}
		\end{equation}
		where \(\operatorname{Part}_{1,2}(m)\) is the set of all partitions of
		\(\{1,\ldots,m\}\) into singletons and pairs.

		For a centered Gaussian environment, moments of an odd number of
		linear forms vanish. For an even number, only complete pairings
		contribute to the sum.
	\end{corollary}

	As expected, for a Gaussian initial state, all field moments required
	at a given perturbative order can be expressed in terms of the first
	and second field moments. The proofs are given in
	Appendix~\ref{app:wick}. To evaluate multitime correlations using
	these formulas, one replaces each \(\vec{\mathbf{x}}_r\) by the
	corresponding vector \(\vec{\mathbf{x}}_r(t_r)\) propagated under
	the free environmental dynamics.

	\section{\label{sec:single-mode-worked-example}A single damped mode in practice}

	We now combine the two parts of the method: propagation of the
	interaction coefficients in the doubled space and evaluation of the
	resulting traces through the initial characteristic function. Consider
	a single bosonic mode with free evolution and amplitude damping,
	\begin{equation}
		\begin{gathered}
		\mathcal L_0^B=-i[\omega_B\hat a^\dagger\hat a ,\bullet]+\gamma\left(
		\hat a\bullet\hat a^\dagger
		-\frac12\{\hat a^\dagger\hat a,\bullet\}
		\right),
		\\
		\omega_B>0,\qquad\gamma\geq0.
		\end{gathered}
	\label{eq:single-mode-generator}
	\end{equation}
	where \(\{\hat X,\hat Y\}=\hat X\hat Y+\hat Y\hat X\).

	The one-mode matrices and operator vectors are
	\begin{equation}
		\mathfrak a=\begin{pmatrix}\hat a\\\hat a^\dagger\end{pmatrix},
		\qquad
		J=\begin{pmatrix}0&-1\\1&0\end{pmatrix},
		\qquad
		E=\begin{pmatrix}0&1\\1&0\end{pmatrix},
	\end{equation}

	The doubled vector is
	\begin{equation}
		\mathfrak A=
		\begin{pmatrix}
			(\hat a\bullet)\\
			(\hat a^\dagger\bullet)\\
			(\bullet\hat a^\dagger)\\
			(\bullet\hat a)
		\end{pmatrix}
	\label{eq:single-mode-doubled-vector}
	\end{equation}

	The generator in Eq.~\eqref{eq:single-mode-generator}
	corresponds to
	\begin{equation}
		H_B=\begin{pmatrix}0&\omega_B\\\omega_B&0\end{pmatrix},
		\qquad
		\Gamma=\begin{pmatrix}0&\gamma\\0&0\end{pmatrix},
	\end{equation}
	Substitution into Eq.~\eqref{eq:doubled-liouvillian-matrix} gives
	\begin{equation}
		\mathbf{L}=\begin{pmatrix}
			0&-\frac{\gamma}{2}-i\omega_B&\gamma&0\\
			-\frac{\gamma}{2}-i\omega_B&0&0&0\\
			\gamma&0&0&-\frac{\gamma}{2}+i\omega_B\\
			0&0&-\frac{\gamma}{2}+i\omega_B&0
		\end{pmatrix}.
		\label{eq:single-mode-liouvillian-matrix}
	\end{equation}

	The doubled propagator in Eq.~\eqref{eq:doubled-linear-propagator}
	places all time dependence in the coefficients of the interaction
	factors. The reduction propositions then express their traces as
	derivatives of the initial characteristic function at the origin.
	Thus the same propagated coefficients can be used for different
	initial states by changing only these derivatives. Damping mixes
	left and right multiplication in the doubled space, and these mixed contributions are retained automatically when
	evaluating products of the evolved superoperators.
	The explicit propagator and the substitutions into the reduction
	formulas are given in Appendix~\ref{app:single-mode-worked-example}.

	\subsection{Moments for linear interaction factors}

	For the first two components of the doubled vector and
	\(0\leq t_1\leq t\), applying
	Proposition~\ref{prop:linear-product-reduction} to their propagated
	coefficients gives
	\begin{align}
		&\operatorname{tr}_B[\mathfrak A_1(t)\mathfrak A_2(t_1)\hat\rho_B]
		\nonumber\\
		&=e^{-(\gamma/2+i\omega_B)(t-t_1)}
		\left[-e^{-\gamma t_1}
		\left.\partial_z\partial_{z^*}\chi\right|_{\vec z=\vec{0}}
		+1-\frac12e^{-\gamma t_1}\right].
		\label{eq:single-mode-linear-two-time}
	\end{align}
	Here \(\chi(\vec z)\) is the characteristic function of the initial
	state, with \(\vec z=(z,z^*)^{\mathsf T}\). The state dependence of
	this correlation is determined by a second derivative of \(\chi\).
	The factor depending on \(t-t_1\) describes the damped oscillation
	between the two interactions, while the dependence on \(t_1\)
	accounts for relaxation before the earlier interaction. Consequently,
	for a general initial state the correlation depends on both times,
	not only on their difference.
	At equal times it gives
	\(e^{-\gamma t}\operatorname{tr}_B[\hat a^\dagger\hat a\hat\rho_B]+1\),
	including the vacuum contribution fixed by operator ordering.

	\subsection{Moments for quadratic interaction factors}

	For quadratic coupling, the same matrix propagation is followed by
	Proposition~\ref{prop:quadratic-product-reduction}. Applied to
	\((\hat a^2\bullet)\) and \(((\hat a^\dagger)^2\bullet)\), it gives
	\begin{align}
		&\operatorname{tr}_B[\mathfrak A_1(t)^2\mathfrak A_2(t_1)^2\hat\rho_B]
		\nonumber\\
		&=e^{-(\gamma+2i\omega_B)(t-t_1)}
		\Bigl[e^{-2\gamma t_1}
		\left.\partial_z^2\partial_{z^*}^{\,2}\chi\right|_{\vec z=\vec{0}}
		\nonumber\\
		&+2e^{-\gamma t_1}(e^{-\gamma t_1}-2)
		\left.\partial_z\partial_{z^*}\chi\right|_{\vec z=\vec{0}}
		+\frac12(e^{-\gamma t_1}-2)^2\Bigr].
		\label{eq:single-mode-quadratic-two-time}
	\end{align}
	The fourth derivative describes the fourth-degree field moment,
	while the second derivative and constant are the corrections due to
	operator ordering. All three terms follow from the same reduction
	formula.

	A single quadratic factor involves derivatives up to second order of the
	characteristic function, as stated in
	Proposition~\ref{prop:quadratic-reduction}. Its product with a second
	quadratic factor also probes fourth derivatives. Thus states with
	identical first and second field moments can yield the same
	first-order reduced evolution but different second-order environmental
	correlations. This is the distinction examined in the two-photon
	Rabi model below. The remaining single-factor and two-factor moments,
	including the right-multiplication channels, are collected in
	Appendix~\ref{app:single-mode-worked-example}.

	\section{\label{sec:two-photon-rabi-model}Environmental statistics in the two-photon Rabi model}
	We now apply the method developed above to a qubit coupled quadratically
	to a damped bosonic mode. The single-mode results of the preceding
	section provide the environmental moments needed to construct its
	reduced dynamics, combining the two-photon interaction with dissipative
	evolution of the mode.
	
	The Hamiltonian of the composite system is that of the two-photon asymmetric Rabi model.
	Two-photon Rabi interactions admit circuit-QED implementations and lead
	to effects such as two-photon blockade and spectral
	collapse~\cite{FelicettiRossattoRicoSolanoFornDiaz2018}.
	The large-eigenvalue asymptotics of the asymmetric model was studied
	in Ref.~\cite{BoutetDeMonvelCharifZielinski2023}. The two-photon spectrum
	undergoes spectral collapse when the coupling reaches one half of the
	bosonic-mode frequency. Including damping allows us to examine how
	these nonlinear interactions transmit the influence of environmental
	dissipation to the qubit. More generally, dissipatively coupled
	oscillator models show how a mode without direct bath coupling can
	acquire effective pumping and relaxation channels through another
	damped mode~\cite{TrubilkoBasharov2020}.
	
	The two-photon interaction also connects this model to the broader
	study of effective higher-order light--matter processes. Effective
	second-order processes have been studied in quantum-stochastic
	descriptions of Stark-modified spontaneous
	emission~\cite{Basharov2011} and collective oscillator
	radiation~\cite{BasharovTrubilko2021}. Non-Wiener corrections have
	also been incorporated into a kinetic equation for an atomic system
	driven by a broadband one-photon
	packet~\cite{TrubilkoBasharov2018}. These works provide related
	physical context, although their methods and perturbative ordering
	differ from those used here. In particular, ``second order'' in
	those works refers to the order in the microscopic coupling used
	to construct an effective interaction. Below, it refers to the
	second-order term of the TCL expansion for the specified
	two-photon interaction.
	
	The model therefore provides a direct
	application of our method to the question of how environmental
	statistics beyond the first two moments affect reduced dynamics.
	Our main comparison is between an initial Fock state and a thermal
	state with the same mean occupation. Their mean vectors and
	covariance matrices coincide, but their fourth moments differ for
	nonzero occupation. 
	
	To make this comparison explicit, we specialize
	Eq.~\eqref{eq:composite-master-equation} to the dissipative
	two-photon asymmetric Rabi model. All model parameters are real: \(\Delta\) denotes the qubit splitting, \(\epsilon\) the bias, \(g\) the coupling constant, \(\omega_B>0\) the environmental mode frequency, and \(\gamma\geq0\) the damping rate.
	The system free generator is
	\begin{equation}
		\mathcal L_0^A
		=
		-i
		\left[
		\frac12
		\left(
		\Delta\hat\sigma_z
		+
		\epsilon\hat\sigma_x
		\right),
		\bullet
		\right],
	\end{equation}
	and the environmental free generator \(\mathcal L_0^B\) is given by
	Eq.~\eqref{eq:single-mode-generator}. We use the one-mode operator
	vectors and generator matrices specified in
	Sec.~\ref{sec:single-mode-worked-example}.
	The interaction generator is
	\begin{equation}
		\mathcal V
		=
		-i
		\left[
		g\hat\sigma_x
		\otimes
		\left(
		\hat a^2+(\hat a^\dagger)^2
		\right),
		\bullet
		\right].
	\end{equation}
	
		The channel decomposition of the interaction is therefore
	\begin{equation}
		\mathcal V
		=
		\mathcal V_1^A\otimes\mathcal V_1^B
		+
		\mathcal V_2^A\otimes\mathcal V_2^B.
	\end{equation}
	
		The corresponding system channels are
	\begin{equation}
		\mathcal V_1^A
		=
		-ig(\hat\sigma_x\,\bullet),
		\qquad
		\mathcal V_2^A
		=
		ig(\bullet\,\hat\sigma_x).
	\end{equation}

	The two environmental interaction matrices are
	\begin{equation}
		\mathbf{V}_1
		=
		\begin{pmatrix}
			I_2&0\\
			0&0
		\end{pmatrix},
		\qquad
		\mathbf{V}_2
		=
		\begin{pmatrix}
			0&0\\
			0&I_2
		\end{pmatrix}.
	\end{equation}
	They define left and right multiplication by the two-photon operator:
	\begin{align}
		\mathcal V_1^B
		&
		=
		\left(\hat a^2+(\hat a^\dagger)^2\right)\bullet
		=
		\mathfrak A^{\mathsf T}\mathbf{V}_1\mathfrak A,
		\\
		\mathcal V_2^B
		&=\bullet\left(\hat a^2+(\hat a^\dagger)^2\right)=
		\mathfrak A^{\mathsf T}\mathbf{V}_2\mathfrak A.
	\end{align}

\subsection{First-order reduced evolution}

The free qubit Hamiltonian
\(\hat H_A=(\Delta\hat\sigma_z+\epsilon\hat\sigma_x)/2\)
determines the interaction-picture operator
\begin{align}
 \hat S(t)
 &=e^{i\hat H_A t}\hat\sigma_x e^{-i\hat H_A t}
 \nonumber\\
 &=\left(\frac{\epsilon^2}{\Omega^2}
 +\frac{\Delta^2}{\Omega^2}\cos\Omega t\right)\hat\sigma_x
 \nonumber\\
 &\quad-\frac{\Delta}{\Omega}\sin\Omega t\,\hat\sigma_y
 +\frac{\epsilon\Delta}{\Omega^2}(1-\cos\Omega t)\hat\sigma_z,
 \label{eq:rabi-system-operator}
\end{align}
where \(\Omega=\sqrt{\Delta^2+\epsilon^2}\). To obtain Eq.~\eqref{eq:rabi-system-operator}, use \(\hat H_A^2=\Omega^2\hat I_A/4\) to write
\begin{equation*}
 e^{i\hat H_A t}=\cos\frac{\Omega t}{2}\,\hat I_A
 +i\frac{\sin(\Omega t/2)}{\Omega}
 (\Delta\hat\sigma_z+\epsilon\hat\sigma_x).
\end{equation*}
Multiplication on both sides of \(\hat\sigma_x\), using \(\hat\sigma_\alpha\hat\sigma_\beta=\delta_{\alpha\beta}\hat I_A+i\sum_\nu\varepsilon_{\alpha\beta\nu}\hat\sigma_\nu\), gives the displayed coefficients. At \(\Omega=0\), the expression is understood by continuity and \(\hat S(t)=\hat\sigma_x\). The interaction-picture system channels are then
\begin{equation}
 \mathcal V_1^A(t)=-ig\bigl(\hat S(t)\,\bullet\bigr),
 \qquad
 \mathcal V_2^A(t)=ig\bigl(\bullet\,\hat S(t)\bigr).
 \label{eq:rabi-system-channels}
\end{equation}

The environmental factor at first order is the one-time correlation
\begin{equation}
 m_j(t)=\operatorname{tr}_B\!\left[
 \mathcal V_j^B\mathcal T_t\hat\rho_B\right],\qquad j=1,2.
\end{equation}
The \(N=1\) case of Eq.~\eqref{eq:environmental-multitime-moment-interaction}
rewrites it as
\begin{equation}
 m_j(t)=\operatorname{tr}_B\!\left[\mathcal V_j^B(t)\hat\rho_B\right].
\end{equation}
For the channels considered here,
\(\mathcal V_j^B=\mathfrak A^{\mathsf T}\mathbf V_j\mathfrak A\).
The similarity transformation preserves products of multiplication
superoperators, and the doubled evolution therefore gives
\begin{equation}
 \mathcal V_j^B(t)=\mathfrak A^{\mathsf T}\mathbf V_j(t)\mathfrak A,
 \qquad
 \mathbf V_j(t)=\mathbf G(t)^{\mathsf T}\mathbf V_j\mathbf G(t).
\end{equation}
The environmental average is consequently the expectation value
of a quadratic form in the reference state. Applying
Proposition~\ref{prop:quadratic-reduction} with
\(\mathbf X=\mathbf V_j(t)\) gives the same result for both channels,
\(m_1(t)=m_2(t)\equiv m(t)\). Their equality also follows from
cyclicity of the trace, since the channels are left and right
multiplication by the same operator. The term without derivatives
vanishes for these channels, so the average depends only on
second derivatives of the characteristic function,
\begin{equation}
 m(t)=-e^{-\gamma t}\left.
 \left(e^{-2i\omega_Bt}\partial_z^2
 +e^{2i\omega_Bt}\partial_{z^*}^{\,2}\right)
 \chi(z,z^*)\right|_{z=z^*=0}.
 \label{eq:first-order-characteristic-average}
\end{equation}
See Appendix~\ref{app:single-mode-worked-example} for details.
Truncating the reduced evolution equation at first order gives
\begin{equation}
 \frac{d}{dt}\hat\rho_A^I(t)
 =-i\lambda g\,m(t)\bigl[\hat S(t),\hat\rho_A^I(t)\bigr].
\end{equation}
This average can also be expressed through the mean vector \(\vec{\mu}\)
and covariance matrix \(\Sigma\):
\begin{equation}
 m(t)=2e^{-\gamma t}\operatorname{Re}\!\left[
 (\mu_1^2+\Sigma_{11})e^{-2i\omega_Bt}\right].
 \label{eq:gaussian-first-order-average}
\end{equation}
The first-order generator therefore cannot distinguish states with identical
first and second field moments. Both the Fock and thermal states considered
below have \(\vec{\mu}=0\) and \(\Sigma_{11}=0\), giving \(m(t)=0\)
and \(\mathcal K_1(t)=0\). Note also that a centered Gaussian state with
\(\Sigma_{11}\neq0\) can give a nonzero first-order contribution, so
Gaussianity alone does not imply that the first-order term vanishes.

We omit the details for subsystem \(A\), since the evaluation of
its contributions is straightforward after vectorizing density
operators by stacking their columns. In this convention,
\begin{align}
	\operatorname{vec}\!\left[\mathcal V_1^A(t)\hat\rho_A\right]
	&=-ig\left[I_2\otimes\hat S(t)\right]\operatorname{vec}(\hat\rho_A),\\
	\operatorname{vec}\!\left[\mathcal V_2^A(t)\hat\rho_A\right]
	&=ig\left[\hat S(t)^{\mathsf T}\otimes I_2\right]\operatorname{vec}(\hat\rho_A).
\end{align}
Compositions of the system channels then become ordinary products
of \(4\times4\) matrices in the same order.

\subsection{Two-photon correlations at second order}

We now consider the second-order environmental averages
\begin{equation}
	m_{l k} ( t , t_1 )
	=
	\operatorname{tr}_B
	\Bigl[
	\mathcal V_l^B
	\mathcal T_{t - t_1}
	\mathcal V_k^B
	\mathcal T_{t_1}
	\hat\rho_B
	\Bigr],
\end{equation}
where \(l,k=1,2\).

Since the trace of left and right multiplication by the final two-photon
operator is the same, \(m_{1k}(t,t_1)=m_{2k}(t,t_1)\).
Each correlation is a sum of the quadratic-factor moments in
Eqs.~\eqref{eq:single-mode-quadratic-two-time}
and~\eqref{eq:single-mode-quadratic-other-moments}. The resulting
expression in derivatives of \(\chi\) is given in
Eq.~\eqref{eq:two-photon-second-order-moments} of
Appendix~\ref{app:rabi-environmental-moments}. It contains derivatives
of second and fourth order, together with the constant ordering terms.
The earlier channel \(k\) determines the ordering-dependent part.

\subsection{\label{sec:matched-fock-thermal-states}Identical covariances, different reduced dynamics}

\emph{Fock state.}
For \(n_B\in\mathbb N_0\) and
\(\hat\rho_B=\lvert n_B\rangle\langle n_B\rvert\),
\begin{equation}
	\chi_{n_B}(z,z^*)
	=
	e^{-zz^*/2}
	L_{n_B}(zz^*),
\end{equation}
where \(L_{n_B}\) is the Laguerre polynomial.

\emph{Thermal state.}
For a thermal state with mean occupation \(n_{\mathrm{th}}\geq0\),
\begin{equation}
	\Sigma_{11}=\Sigma_{22}=0,
	\qquad
	\Sigma_{12}=n_{\mathrm{th}}+\frac12,
	\qquad
	\vec{\mu}=\vec{0},
\end{equation}
and hence
\begin{equation}
	\chi_{\mathrm{th}}(z,z^*)
	=
	\exp\left[
	-\left(n_{\mathrm{th}}+\frac12\right)zz^*
	\right].
\end{equation}

\emph{Comparison at equal mean occupation.}
Set \(n_{\mathrm{th}}=n_B\). Both states then have
\begin{equation}
	\vec{\mu}=0,
	\qquad
	\Sigma=
	\begin{pmatrix}
		0&n_B+\frac12\\
		n_B+\frac12&0
	\end{pmatrix},
\end{equation}
The required mixed derivatives
are obtained in Appendix~\ref{app:matched-fourth-derivatives} from the
arbitrary-order formulas for the two classes of states.
The contribution responsible for the difference in the reduced
dynamics is

\begin{equation}
	\left.
	\partial_z^2\partial_{z^*}^{\,2}
	\bigl(\chi_{\mathrm{th}}-\chi_{n_B}\bigr)
	\right|_{\vec{z}=\vec{0}}
	=n_B(n_B+1).
	\label{eq:fock-thermal-fourth-difference}
\end{equation}
At \(n_B=0\), both states reduce to the vacuum, so the difference
vanishes. For nonzero occupation, the thermal state remains Gaussian,
although it is not stationary under the zero-temperature damping
considered here.

The full expression for \(m_{lk}(t,t_1)\) is lengthy and is given
in Appendix~\ref{app:two-photon-rabi-details},
Eq.~\eqref{eq:two-photon-second-order-moments}.
However, the difference between the thermal and Fock correlations
takes a simple form. Substituting
Eq.~\eqref{eq:fock-thermal-fourth-difference} into that expression
gives, for \(0\leq t_1\leq t\) and any \(l,k\in\{1,2\}\),
\begin{equation}
 \begin{aligned}
 &m_{lk}^{\mathrm{th}}(t,t_1)-m_{lk}^{\mathrm F}(t,t_1)
 \\
 &\quad=2n_B(n_B+1)e^{-\gamma(t+t_1)}
 \cos\!\left[2\omega_B(t-t_1)\right].
 \end{aligned}
 \label{eq:fock-thermal-correlation-difference}
\end{equation}
The superscripts specify the initial thermal or Fock state. The difference is real and independent of the channels, so the imaginary parts of the corresponding correlation functions coincide. Denote the common difference in Eq.~\eqref{eq:fock-thermal-correlation-difference} by \(\delta m(t,t_1)\). Since the first-order generators vanish and the sum of the system channels in Eq.~\eqref{eq:rabi-system-channels} is \(-ig[\hat S(t),\bullet]\), substitution into the second-order expression in Sec.~\ref{sec:moment-functions} gives
\begin{equation}
 \begin{aligned}
 &\bigl(\mathcal K_2^{\mathrm{th}}(t)-\mathcal K_2^{\mathrm F}(t)\bigr)\hat\rho_A
 \\
 &\quad=-g^2\int_0^t dt_1\,\delta m(t,t_1)
 [\hat S(t),[\hat S(t_1),\hat\rho_A]].
 \end{aligned}
 \label{eq:fock-thermal-generator-difference}
\end{equation}
Although the thermal state has larger initial fluctuations of the
quadratic interaction operator, the oscillatory time dependence
of \(\delta m(t,t_1)\) does not necessarily lead to stronger damping
at every time in the reduced dynamics.

\raggedbottom
\section{\label{sec:numerical-comparison}Fock and thermal environments: dynamics and accuracy}

We now compare the TCL2 and TCL4 approximations with the reduced numerical solution of the composite master equation for the initial Fock and thermal states considered above on \(0\leq t\leq2/\gamma\).

For both states, all odd perturbative orders vanish.
Indeed, the phase rotation
\(\hat U_B=e^{i\pi\hat a^\dagger\hat a/2}\) leaves each state invariant,
commutes with the free environmental evolution under conjugation, and
changes the sign of \(\hat a^2+(\hat a^\dagger)^2\). Every ordered multitime environmental correlation function
with an odd number of interaction factors is therefore zero.
Consequently, the moments and the recurrence
\eqref{eq:cumulant-recurrence} give
\(\mathcal M_{2r+1}(t)=\mathcal K_{2r+1}(t)=0\) for \(r\geq0\).
In particular,
\begin{align}
	\mathcal K_2(t)&=\dot{\mathcal M}_2(t),
	\\
	\mathcal K_4(t)&=\dot{\mathcal M}_4(t)
	-\dot{\mathcal M}_2(t)\mathcal M_2(t).
	\label{eq:matched-even-tcl-generators}
\end{align}

The TCL2 and TCL4 approximations used below are obtained by truncating
the generator in Eq.~\eqref{eq:matched-even-tcl-generators} at second
and fourth order. The TCL2 approximation solves
\begin{equation*}
	\frac{d}{dt}\hat\rho_A^I(t)
	=\lambda^2\mathcal K_2(t)\hat\rho_A^I(t),
\end{equation*}
while TCL4 includes the fourth-order contribution and solves
\begin{equation*}
	\frac{d}{dt}\hat\rho_A^I(t)
	=\left[\lambda^2\mathcal K_2(t)+\lambda^4\mathcal K_4(t)\right]
	\hat\rho_A^I(t).
\end{equation*}

The composite master equation is solved numerically using QuTiP, and the qubit state is obtained by tracing out the environmental mode. The reduced TCL2 and TCL4 generators are constructed algebraically in Wolfram Mathematica, which is also used to solve the resulting evolution equations numerically.

To implement the differential representation, we identify derivatives of
the delta function with monomials in auxiliary variables
\(\vec u=(u_1,u_2)^{\mathsf T}\):
\begin{equation*}
	i^{p+q}\partial_z^p\partial_{z^*}^{\,q}\delta^{(2)}(\vec z)
	\quad\longleftrightarrow\quad u_1^p u_2^q.
\end{equation*}
Under this identification, the differential operator in
Eq.~\eqref{eq:def-D-script-x} is represented by
\begin{equation*}
	\mathscr L
	=\vec p^{\mathsf T}\vec u
	+\frac12\vec q^{\mathsf T}J\partial_{\vec u},
\end{equation*}
where \(\vec p\) and \(\vec q\) are the coefficient combinations
defined above for the corresponding linear form.

In a finite monomial basis whose first element \(1\) corresponds to
the delta function, multiplication and differentiation are represented
by matrices. Denoting the coefficient vector of the first monomial by
\(\overrightarrow{\mathrm{mon}}_1=(1,0,\ldots,0)^{\mathsf T}\) and the channel
matrices by \(C_i(t)\), we obtain the ordered correlation as
\begin{equation}
	m_{i_n\dots i_1}(t_n,\ldots,t_1)
	=\vec m_\chi^{\mathsf T}
	C_{i_1}(t_1)\cdots C_{i_n}(t_n)\overrightarrow{\mathrm{mon}}_1,
	\label{eq:environmental-correlation-matrix-product}
\end{equation}
where \(0\leq t_1\leq\cdots\leq t_n\), and \(\vec m_\chi\) contains
the reference-state moments in the same monomial basis,
\begin{equation*}
	(\vec m_\chi)_{pq}
	=(-i)^{p+q}\left.
	\partial_z^p\partial_{z^*}^{\,q}\chi(\vec z)
	\right|_{\vec z=\vec 0},
\end{equation*}
where the components of \(\vec m_\chi\) are indexed by the same
exponent pairs \((p,q)\) as the monomial basis and contain the
corresponding reference-state moments,

Combining these correlations with the system-channel matrices gives a
finite sum of exponentials in the ordered time variables. The repeated
integrals are evaluated using the triangular-matrix exponential formula
given in Ref.~\cite{CarbonellJimenezPedroso2008}.
For each vector of exponential rates \(\vec\nu=(\nu_1,\ldots,\nu_k)\),
\begin{equation}
	\int_{0<t_1<\cdots<t_k<t}
	e^{\sum_{j=1}^k\nu_jt_j}\,dt_1\cdots dt_k
	=\left[e^{tF_{\vec\nu}}\right]_{1,k+1},
	\label{eq:ordered-integral-matrix-exponential}
\end{equation}
where \(F_{\vec\nu}\) is upper bidiagonal, with diagonal entries
\(0,\sum_{j=1}^k\nu_j,\sum_{j=2}^k\nu_j,\ldots,\nu_k\)
and ones on the first superdiagonal. This formula also handles coincident
diagonal entries, yielding the corresponding polynomial factors in time
without singular denominators. Further implementation details and
reproduction scripts are provided in Ref.~\cite{QuadraticBathTCLCode}.

To choose the coupling strength for these comparisons, we use a diagnostic motivated by weak-coupling estimates. For Gaussian environmental correlations, as arise from linear coupling to a Gaussian bosonic environment, error bounds for the Redfield, Davies, and coarse-grained master equations are derived in Ref.~\cite{MozgunovLidar2020} under the assumptions stated there. A central small parameter is $\tau_B/\tau_{SB}\ll1$, where $\tau_B$ is the bath correlation time and $\tau_{SB}$ is the characteristic system relaxation timescale, both defined through integrals of the bath correlation function. For an exponentially decaying correlation, this ratio scales as the squared product of the effective coupling amplitude and the correlation time.
Explicit bounds on the residual term after finite-order Born and Markov
approximations for Gaussian environments are obtained in
Ref.~\cite{AgerskovNathan2026}. Under suitable bounds on the time moments of the
correlation matrices, an appropriate finite truncation order, increasing
as the coupling decreases, yields a residual smaller than any fixed
power of the weak-coupling parameter.

Comparisons of second- and fourth-order contributions also provide practical weak-coupling criteria for systems linearly coupled to thermal bosonic baths~\cite{FruchtmanLambertGauger2016}. For the present nonstationary environment, quadratic interaction factors need not have Gaussian statistics even when the initial field state is Gaussian, so those criteria do not apply directly. We use the coupling-amplitude--memory-time scaling to motivate a numerical diagnostic.

To select the coupling strength for a simulation interval $[0,T]$ with
$\gamma>0$, we use the Fock-based diagnostic
\begin{equation}
 \eta_{\mathrm{num}}(T)
 =|\lambda g|\sqrt{2(n_B^2+n_B+1)}
   \frac{1-e^{-\gamma T}}{\gamma}.
 \label{eq:numerical-coupling-diagnostic}
\end{equation}
The condition $\eta_{\mathrm{num}}(T)\ll1$ guides parameter selection but does not bound the TCL truncation error, since it does not control higher connected environmental correlations. We assess accuracy against the numerical solution of the composite master equation. The envelope estimate and its zero-damping limit are derived in Appendix~\ref{app:rabi-coupling-diagnostic}.

\def\articleMainEta{0.5}
\def\articleMainG{0.000567029}
\def\articleEtaValues{0.35,0.4,0.45,0.5,0.55}

For the main illustrations, we set
\(\eta\equiv\eta_{\mathrm{num}}(2/\gamma)=\articleMainEta\)
using the diagnostic in Eq.~\eqref{eq:numerical-coupling-diagnostic}.
This deliberately less conservative choice makes the differences
between the numerical solution, TCL2, and TCL4 distinguishable
in the plots. The same coupling strength is used for both
environmental states.

\paragraph{Common parameters.}\label{par:article-parameters}
Parameters: $\omega_B=1$ (frequency unit), $\Delta=1.985$, $\epsilon=0$, $\gamma=0.005$ and $\lambda=1$. Initially $\rho_{AB}=\rho_A\otimes\rho_B$, $\rho_A=[I+r_x\sigma_x+r_y\sigma_y+r_z\sigma_z]/2$, $(r_x,r_y,r_z)=(0.6,0,0.8)$. The initial mode states are Fock, $|n_B\rangle\langle n_B|$, and undisplaced, unsqueezed thermal, with $n_B=\bar n=3$ . Mode damping is at zero temperature. The observation interval is $0\leq\gamma t\leq2$, with 1000 shared output times. The numerical oscillator cutoffs are 40 (Fock) and 73 (thermal).

For each environment and coupling \(\eta\), we quantify accuracy by the trace distance from the reduced numerical solution and its maximum over the observation interval:
\begin{equation*}
 \begin{aligned}
 D_k(t)&=\frac12\bigl\|\hat\rho_{\mathrm{num}}^I(t)-\hat\rho_{\mathrm{TCL}k}^I(t)\bigr\|_1,\\
 E_k(\eta)&=\max_{0\leq t\leq2/\gamma}D_k(t),\qquad k=2,4.
 \end{aligned}
\end{equation*}

For all displayed TCL2 and TCL4 trajectories, positivity, trace
normalization, and Hermiticity were preserved within numerical
precision at the sampled times. These checks concern the chosen
initial state and do not establish complete positivity of the
approximate dynamical maps.

Throughout the numerical figures, the notation
\begin{equation}
 \langle\sigma_\alpha(t)\rangle
 :=\operatorname{Tr}_A\!\Bigl[\hat\rho_A^I(t)\hat\sigma_\alpha\Bigr],
 \qquad \alpha=x,y,z,
 \label{eq:plotted-pauli-means}
\end{equation}
denotes Pauli expectation values in the interaction picture. Here
$\hat\rho_A^I(t)=e^{i\hat H_A t}\hat\rho_A(t)e^{-i\hat H_A t}$, so the free
qubit rotation is removed.

Figure~\ref{fig:article-combined-bloch-main} compares the
interaction-picture Pauli expectation values for the Fock and thermal
environments, while Fig.~\ref{fig:article-combined-purity-main}
shows the corresponding purities. Together, these figures illustrate
how identical initial field covariances can lead to different reduced
dynamics, consistent with the second-order contribution in
Eq.~\eqref{eq:fock-thermal-generator-difference}.
\onecolumngrid
\articlefigure{\textwidth}{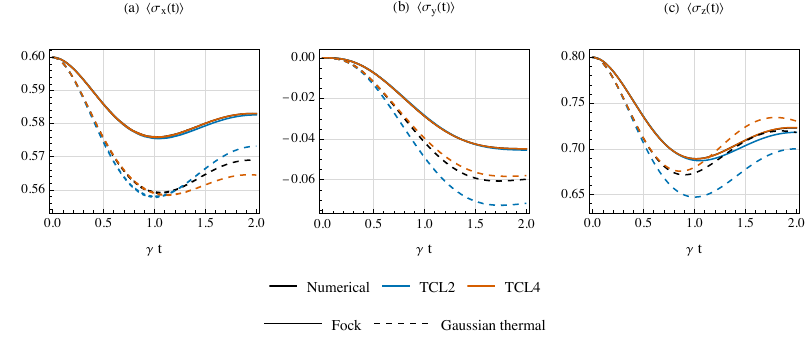}{%
\label{fig:article-combined-bloch-main}Interaction-picture Pauli expectation values (a) $\langle\sigma_x(t)\rangle$, (b) $\langle\sigma_y(t)\rangle$, and (c) $\langle\sigma_z(t)\rangle$ at $\eta=\articleMainEta$ and $g=\articleMainG$. The numerical solution is black, TCL2 blue, and TCL4 orange. Solid curves correspond to the Fock environment and dashed curves to the thermal environment. Parameters and initial states are specified in Sec.~\ref{sec:numerical-comparison}.}
\twocolumngrid

\articlefigure{\linewidth}{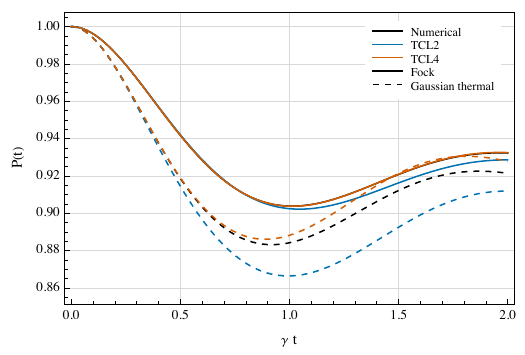}{%
\label{fig:article-combined-purity-main}Purity $P(t)=\Tr\rho^2(t)$ for the same environments, coupling, initial states and curve conventions as in Fig.~\ref{fig:article-combined-bloch-main}.}

Figure~\ref{fig:article-error-scan} shows how the maximum
trace-distance error varies with coupling strength.
TCL4 improves accuracy for both environments at every displayed
coupling. At \(\eta=\articleMainEta\), the error decreases relative
to TCL2 by approximately a factor of 20 for the Fock environment
and a factor of 2 for the thermal environment.
\articlefigure{\linewidth}{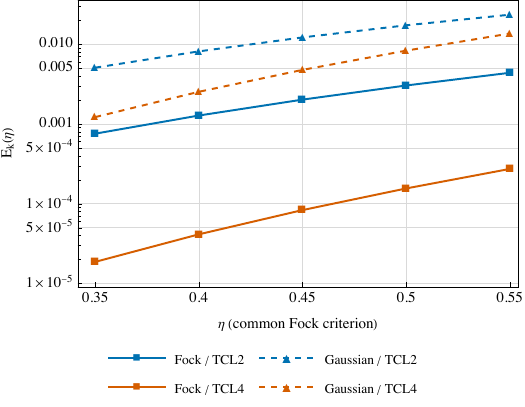}{%
\label{fig:article-error-scan}Maximum trace-distance errors $E_2(\eta)$ and $E_4(\eta)$ for the Fock and thermal environments. The maxima are evaluated on the sampled time grid.}

\section{Conclusions}

We have developed a perturbative method for the reduced dynamics of a system coupled quadratically to a bosonic environment with a quadratic Hamiltonian and linear Lindblad operators. For a factorized initial state, the method accommodates non-Gaussian and nonstationary environmental states with the required finite moments. Linear evolution of multiplication superoperators in a doubled space, combined with a characteristic-function representation, reduces ordered multitime environmental correlations to finite sums of derivatives of the initial characteristic function. These correlations supply the moments entering the recursive TCL expansion at any prescribed finite order. The construction applies to multimode environments and separates the propagation of the interaction coefficients from the specification of the initial environmental state.

The two-photon Rabi model with a damped environmental mode illustrates how quadratic coupling probes statistics beyond the field covariance. Initial Fock and thermal states with the same mean occupation have identical first and second field moments, but their fourth moments differ for nonzero occupation. We obtained the resulting difference between their two-time interaction correlations explicitly. The first-order reduced generators vanish for both states, whereas the second-order generators distinguish them. More generally, perturbative order \(N\) requires field moments up to degree \(2N\). Comparisons with the numerical solution of the composite master equation show that including the fourth-order TCL contribution reduces the maximum trace-distance error for both initial states over the coupling range displayed here.

The numerical benchmarks test finite-order accuracy in the regimes considered, without establishing convergence of the perturbative series. Deriving truncation-error bounds for nonstationary environments with nonlinear coupling remains an open problem, for which estimates of higher connected multitime correlations are needed.

\begin{acknowledgments}
The authors thank Kh.\,Sh. Meretukov and R. Singh for fruitful discussions
of the problems considered in this work. 

This work was supported by the Russian Science Foundation under grant no. 25-21-00700, \url{https://rscf.ru/project/25-21-00700/}.
\end{acknowledgments}

\bibliography{gameOfNons}

	\clearpage
	\raggedbottom
	\appendix

	\section{\label{app:doubled-heisenberg}Heisenberg evolution of multiplication superoperators}

	In this appendix we collect the technical results underlying the doubled-space
	formalism used in the main text.

	\subsection{Canonical commutation relations}

	We use the doubled vector from Eq.~\eqref{eq:doubled-vector} and the
	matrices \(E,J\) from Eq.~\eqref{eq:exchange-matrix}. The original
	field commutator is \([\mathfrak a,\mathfrak a^{\mathsf T}]=-J\).

\begin{proposition}
	The doubled vector of multiplication superoperators satisfies
	\begin{equation}
		[\mathfrak A,\mathfrak A^{\mathsf T}]
		=
		-\mathbf{J},
	\end{equation}
	where \(\mathbf{J} =J\oplus J.\)
\end{proposition}

\begin{proof}
	For left multiplication one has
	\begin{equation}
		\left[
		(\mathfrak a \, \bullet)_i,
		(\mathfrak a \, \bullet)_j
		\right]
		=
		([\mathfrak a_i,\mathfrak a_j]\,\bullet)
		=
		-J_{ij}\mathcal I_B.
	\end{equation}

	For the right block we first note that
	\begin{equation}
		\left[
		(E \mathfrak a)_i,
		(E \mathfrak a)_j
		\right]
		=
		\left(
		E
		[\mathfrak a,\mathfrak a^{\mathsf T}]
		E^{\mathsf T}
		\right)_{ij}
		=
		(-EJE^{\mathsf T})_{ij}
		=
		J_{ij}.
	\end{equation}
	Since right multiplication reverses the order of operators,
	\begin{equation}
		\left[
		(\bullet \, E \mathfrak a)_i,
		(\bullet \, E \mathfrak a)_j
		\right]
		=
		-\bigl(\bullet\,[(E\mathfrak a)_i,(E\mathfrak a)_j]\bigr)
		=
		-J_{ij}\mathcal I_B.
	\end{equation}

	Moreover, left and right multiplication commute,
	\begin{equation}
		\left[
		(\mathfrak a \, \bullet)_i,
		(\bullet \, E \mathfrak a)_j
		\right]
		=
		0.
	\end{equation}
	Therefore,
	\begin{eqnarray}
		[\mathfrak A,\mathfrak A^{\mathsf T}]
		=
		\begin{pmatrix}
			-J&0\\
			0&-J
		\end{pmatrix} = -\mathbf{J}
	\end{eqnarray}
\end{proof}

	\subsection{Heisenberg evolution of multiplication superoperators}

	The symmetric coefficient matrix in Eq.~\eqref{eq:doubled-liouvillian} is
	\begin{equation}
		\mathbf{L}
		=
		\begin{pmatrix}
			-iH-\dfrac{\Gamma+\Gamma^{\mathsf T}}2
			&
			\Gamma E
			\\[1ex]
			E\Gamma^{\mathsf T}
			&
			iEHE
			-
			E\dfrac{\Gamma+\Gamma^{\mathsf T}}2E
		\end{pmatrix}.
		\label{eq:doubled-liouvillian-matrix}
	\end{equation}

	For the evolution defined in Eq.~\eqref{eq:doubled-evolution-definition},
	we obtain the following result.

	\begin{proposition}
		The doubled vector of multiplication superoperators satisfies
		\begin{equation}
			\frac{d}{dt}
			\mathfrak A(t)
			=
			-\mathbf{J}\mathbf{L}\,
			\mathfrak A(t).
		\end{equation}
		Consequently,
		\begin{equation}
			\mathfrak A(t)
			=
			\mathbf{G}(t)\mathfrak A,
		\end{equation}
		where
		\begin{equation}
			\mathbf{G}(t)
			=
			e^{-\mathbf{J}\mathbf{L}t}.
		\end{equation}
	\end{proposition}

	\begin{proof}
		Differentiating Eq.~\eqref{eq:doubled-evolution-definition} gives
		\begin{equation}
			\frac{d}{dt}\mathfrak A(t)
			=
			\left[
			\mathfrak A(t),
			\frac12
			\mathfrak A(t)^{\mathsf T}
			\mathbf{L}
			\mathfrak A(t)
			\right].
		\end{equation}

		Using
		\[
		[\mathfrak A(t),\mathfrak A(t)^{\mathsf T}]
		=
		-\mathbf{J},
		\]
		one finds
		\begin{equation}
			\frac{d}{dt}
			\mathfrak A(t)
			=
			-\mathbf{J}\mathbf{L}\,
			\mathfrak A(t).
		\end{equation}

		The stated expression for \(\mathfrak A(t)\) is the solution with
		\(\mathfrak A(0)=\mathfrak A\).
	\end{proof}

Consequently, for any \(\vec{\mathbf{x}}\in\mathbb C^{4d}\),
\begin{equation}
	\vec{\mathbf{x}}^{\mathsf T}\mathfrak A(t)
	=
	\vec{\mathbf{x}}(t)^{\mathsf T}\mathfrak A,
	\qquad
	\vec{\mathbf{x}}(t)
	=
	\mathbf{G}(t)^{\mathsf T}\vec{\mathbf{x}}.
\end{equation}
Since the similarity transformation defining \(\mathfrak A(t)\)
preserves products, for any symmetric matrix
\(\mathbf{X}\in\mathbb C^{4d\times4d}\) we likewise obtain
\begin{equation}
	\mathfrak A(t)^{\mathsf T}\mathbf{X}\mathfrak A(t)
	=
	\mathfrak A^{\mathsf T}\mathbf{X}(t)\mathfrak A,
	\qquad
	\mathbf{X}(t)
	=
	\mathbf{G}(t)^{\mathsf T}\mathbf{X}\mathbf{G}(t).
\end{equation}
Thus homogeneous linear and quadratic forms retain their form
under the free evolution, with time dependence carried by their
coefficients.

	\section{\label{app:differential-representation}Differential representation in the original phase space}

	This appendix specifies the phase-space conventions and proves the
	differential identities used in the main text.

	\subsection{\label{app:weyl-conventions}Weyl representation and phase-space conventions}

	Integration over
	\(\mathbb C^d\simeq\mathbb R^{2d}\) is performed with respect to the standard
	Lebesgue measure
	\begin{equation}
		\mathrm d^{2d}\vec{z}
		:=
		\prod_{k=1}^d
		\mathrm d\operatorname{Re}z_k\,
		\mathrm d\operatorname{Im}z_k.
	\end{equation}

	We denote the corresponding Dirac distribution by
	\begin{equation}
		\delta^{(2d)}(\vec{z})
		=
		\prod_{k=1}^d
		\delta\!\left(\operatorname{Re}z_k\right)
		\delta\!\left(\operatorname{Im}z_k\right).
	\end{equation}

For the characteristic function defined in Sec.~\ref{sec:quadratic-environment},
the inverse relation takes the form
	\begin{equation}
		\hat\rho_B
		=
		\frac1{\pi^d}
		\int
		\mathrm d^{2d}\vec{z}\,
	\chi(\vec{z})\,
	e^{-i\vec{z}^{\mathsf T}\mathfrak a}.
	\label{eq:weyl-inversion}
\end{equation}

	\subsection{Differentiation and distributions}

	The Wirtinger derivatives are defined by
	\begin{equation}
\begin{aligned}
		\partial_{z_k}
		&=
		\frac12
		\left(
		\partial_{\operatorname{Re}z_k}
		-i\partial_{\operatorname{Im}z_k}
		\right),
		\\
		\partial_{z_k^*}
		&=
		\frac12
		\left(
		\partial_{\operatorname{Re}z_k}
		+i\partial_{\operatorname{Im}z_k}
		\right).
	\end{aligned}
\end{equation}
	For a multi-index \(\alpha\in\mathbb N_0^{2d}\), derivatives of the Dirac
	distribution are defined by
	\begin{equation}
		\int \mathrm d^{2d}\vec{z}\,
		f(\vec{z})
		\partial_{\vec{z}}^{\alpha}
		\delta^{(2d)}(\vec{z})
		=
		(-1)^{|\alpha|}
		\left.
		\partial_{\vec{z}}^{\alpha}f(\vec{z})
		\right|_{\vec{z}=\vec{0}}.
	\end{equation}

	\subsection{Left and right actions on Weyl operators}

	\begin{proposition}
		For every vector \(\vec{h}\in\mathbb C^{2d}\),
		\begin{align}
			\vec{h}^{\mathsf T}\mathfrak a\,
			\hat W(\vec{z})
			&=
			i \vec{h}^{\mathsf T}
			\left(
			\partial_{\vec{z}}
			+
			\frac12J\vec{z}
			\right)
			\hat W(\vec{z}),
			\\
			\hat W(\vec{z})\,
			\vec{h}^{\mathsf T}\mathfrak a
			&=
			i \vec{h}^{\mathsf T}
			\left(
			\partial_{\vec{z}}
			-
			\frac12J\vec{z}
			\right)
			\hat W(\vec{z}).
		\end{align}
	\end{proposition}

	\begin{proof}
		The first identity is the standard derivative formula for a Weyl
		operator. For the second, set
		\(
		\hat A=-i\vec{z}^{\mathsf T}\mathfrak a
		\).
		The canonical commutation relations give
		\(
		[\hat A,\mathfrak a]=-iJ\vec{z}
		\), and all higher commutators vanish. Hence
		\begin{equation}
			\hat W(\vec{z})\mathfrak a
			=
			\left(
			\mathfrak a-iJ\vec{z}
			\right)\hat W(\vec{z}).
		\end{equation}
		Substitution of the first identity proves the second.
	\end{proof}

	\subsection{Linear forms in the doubled space}
	For a doubled vector \(\vec{\mathbf{x}}=(\vec{x}_L,\vec{x}_R)^{\mathsf T}\), the linear form is
	\begin{equation}
		\vec{\mathbf{x}}^{\mathsf T}\mathfrak A
		=
		\vec{x}_L^{\mathsf T}\mathfrak a\,\bullet
		+
		\bullet\,\vec{x}_R^{\mathsf T}E\mathfrak a
	\end{equation}
	Applying the preceding proposition to the two blocks gives
	\((\vec{\mathbf{x}}^{\mathsf T}\mathfrak A)\hat W(\vec{z})
	=\mathscr D_{\vec{\mathbf{x}}}\hat W(\vec{z})\), with
	\(\mathscr D_{\vec{\mathbf{x}}}\) defined in Eq.~\eqref{eq:def-D-script-x}.
	Since differential operators commute with left and right multiplication
	by fixed environmental operators, iteration reverses the order of the
	factors. Taking the trace and using
	\(\operatorname{tr}_B[\hat W(\vec{z})]=\pi^d\delta^{(2d)}(\vec{z})\)
	then gives Eq.~\eqref{eq:ordered-doubled-string}.

	Substitution of the inverse Weyl transform,
	Eq.~\eqref{eq:weyl-inversion}, into the definition of
	\(m(\vec{\mathbf{x}}_1,\ldots,\vec{\mathbf{x}}_m)\) cancels the factor \(\pi^d\) and gives
	Eq.~\eqref{eq:environmental-moment-differential}.

	Finally, the rank-one decomposition in Eq.~\eqref{eq:rank-one-decomposition}
	gives
	\(\mathfrak A^{\mathsf T}\mathbf{X}\mathfrak A
	=\sum_\alpha(\vec{\mathbf{x}}_\alpha^{\mathsf T}\mathfrak A)
	(\vec{\mathbf{y}}_\alpha^{\mathsf T}\mathfrak A)\).
	Applying Eq.~\eqref{eq:ordered-doubled-string} term by term yields the
	quadratic-form identity in the main text. Expanding each quadratic factor
	in an ordered product in the same way gives its reduction to the moments
	\(m\), without changing the order of the original factors.

	The explicit evaluation of these moments by differential normal ordering
	and integration by parts is given in Appendix~\ref{app:wick}.

	\section{\label{app:wick}Algebraic reduction of environmental moments}

	\subsection{Differential normal ordering and environmental moments}

	We first establish the differential identity used to prove
	Proposition~\ref{prop:environmental-moments}.

	\begin{lemma}[Differential normal ordering]
		\label{lem:wick-delta}
		For arbitrary \(\vec{p}_i,\vec{q}_i\in\mathbb C^{2d}\), let
		\[
		\mathscr D_i
		=
		i
		\left(
		\vec{p}_i^{\mathsf T}\partial_{\vec{z}}
		+
		\frac12\vec{q}_i^{\mathsf T}J\vec{z}
		\right),
		\qquad
		i=1,\dots,m.
		\]
		Then, in the sense of distributions,
		\begin{align}
			&
			\mathscr D_1
			\cdots
			\mathscr D_m
			\delta^{(2d)}(\vec{z})
			\nonumber\\
			&=
			\sum_{\pi\in\operatorname{Pair}_{\leq}(m)}
			\left(
			\prod_{\substack{\{r,s\}\in\pi\\ r<s}}
			\frac12
			\vec{q}_r^{\mathsf T}
			J
			\vec{p}_s
			\right)
			\nonumber\\
			&\qquad{}\times
			\left(
			\prod_{k\notin\operatorname{supp}\pi}
			i \vec{p}_k^{\mathsf T}\partial_{\vec{z}}
			\right)
			\delta^{(2d)}(\vec{z}),
		\end{align}
		where \(\operatorname{Pair}_{\leq}(m)\) is the set of all partial
		pairings of \(\{1,\dots,m\}\). The derivative operators associated
		with unpaired indices commute.
	\end{lemma}

	\begin{proof}[Proof of Lemma~\ref{lem:wick-delta}]
		The proof proceeds by induction over \(m\).

		For \(m=0\), the statement is trivial. For \(m=1\), using
		\(
		\vec{z}\,\delta^{(2d)}(\vec{z})=0
		\),
		we obtain
		\[
		\mathscr D_1\delta^{(2d)}(\vec{z})
		=
		i \vec{p}_1^{\mathsf T}\partial_{\vec{z}}
		\delta^{(2d)}(\vec{z}),
		\]
		which is the required expression. For the induction step, assume the
		formula for \(\mathscr D_1\cdots\mathscr D_m\delta^{(2d)}(\vec{z})\)
		and apply \(\mathscr D_0\) from the left. Use the distributional identity
		\[
\begin{aligned}
&
		z_i
		\partial_{\vec{z},j_1}
		\cdots
		\partial_{\vec{z},j_l}
		\delta^{(2d)}(\vec{z})
		\\
		&=
		-
		\sum_{\alpha=1}^l
		\delta_{i j_\alpha}
		\partial_{\vec{z},j_1}
		\cdots
		\widehat{\partial_{\vec{z},j_\alpha}}
		\cdots
		\partial_{\vec{z},j_l}
		\delta^{(2d)}(\vec{z}).
		\end{aligned}
\]
		The derivative part of \(\mathscr D_0\) adds the unpaired factor
		\(i \vec{p}_0^{\mathsf T}\partial_{\vec{z}}\). Its coordinate part contracts
		index \(0\) with one previously unpaired index \(k\). The displayed
		identity and the factors of \(i\) give the coefficient
		\(\vec{q}_0^{\mathsf T}J\vec{p}_k/2\). These are exactly all partial pairings after
		adding index \(0\).
	\end{proof}

	\begin{proof}[Proof of Proposition~\ref{prop:environmental-moments}]
		By Eq.~\eqref{eq:environmental-moment-differential}, the ordered
		moment is the integral of \(\chi(\vec{z})\) against
		\(\mathscr D_{\vec{\mathbf{x}}_m}\cdots\mathscr D_{\vec{\mathbf{x}}_1}
		\delta^{(2d)}(\vec{z})\). Apply Lemma~\ref{lem:wick-delta}
		to this reversed differential-operator product. For a pair of
		original indices \(r<s\), the resulting contraction is
		\(\vec{q}_{\vec{\mathbf{x}}_s}^{\mathsf T}J\vec{p}_{\vec{\mathbf{x}}_r}/2\).
		Each unpaired index contributes
		\(i \vec{p}_{\vec{\mathbf{x}}_k}^{\mathsf T}\partial_{\vec{z}}\) acting on the
		distribution. Integrating by parts transfers all
		\(m-2|\pi|\) derivatives to \(\chi\), so their combined factor
		becomes \((-i)^{m-2|\pi|}\). The remaining Dirac distribution
		evaluates these derivatives at \(\vec{z}=0\), yielding the
		stated moment formula.
	\end{proof}

	\subsection{Gaussian environmental moments}

		The quantities \(\mu(\vec{\mathbf{x}}_r)\) and \(\kappa(\vec{\mathbf{x}}_r,\vec{\mathbf{x}}_s)\) defined in Corollary~\ref{cor:gaussian-environmental-moments} are respectively the mean and the ordered centered two-point function of \(\vec{\mathbf{x}}_r^{\mathsf T}\mathfrak A\) and \(\vec{\mathbf{x}}_s^{\mathsf T}\mathfrak A\):
		\begin{align}
			&\mu(\vec{\mathbf{x}}_r)=
			\operatorname{tr}_B
			\Bigl[
			(\vec{\mathbf{x}}_r^{\mathsf T}\mathfrak A)\hat\rho_B
			\Bigr],
			\\
			&\kappa(\vec{\mathbf{x}}_r,\vec{\mathbf{x}}_s)=
			\operatorname{tr}_B
			\Bigl[
			(\vec{\mathbf{x}}_r^{\mathsf T}\mathfrak A-\mu(\vec{\mathbf{x}}_r)\mathcal I_B)
			(\vec{\mathbf{x}}_s^{\mathsf T}\mathfrak A-\mu(\vec{\mathbf{x}}_s)\mathcal I_B)
			\hat\rho_B
			\Bigr].
		\end{align}

	\begin{proof}[Proof of Corollary~\ref{cor:gaussian-environmental-moments}]
		By Proposition~\ref{prop:environmental-moments},
		differentiation of the Gaussian characteristic function shows that every remaining index is
		either a singleton with weight \(\vec{p}_{\vec{\mathbf{x}}_r}^{\mathsf T}\vec{\mu}\), or belongs
		to a pair with weight \(\vec{p}_{\vec{\mathbf{x}}_s}^{\mathsf T}\Sigma \vec{p}_{\vec{\mathbf{x}}_r}\). For any fixed
		pair, the latter state contribution and the differential contraction
		are mutually exclusive alternatives. Summing the two alternatives
		gives \(\kappa(\vec{\mathbf{x}}_r,\vec{\mathbf{x}}_s)\), and summing over all singleton--pair partitions
		proves Eq.~\eqref{eq:environmental-wick}.
	\end{proof}

		In particular, for a centered Gaussian environment moments of an odd
		number of linear forms \(\vec{\mathbf{x}}_r^{\mathsf T}\mathfrak A\) vanish, while
		for even \(m\)
		\begin{equation}
			m(\vec{\mathbf{x}}_1,\ldots,\vec{\mathbf{x}}_m)
			=
			\sum_{\pi\in\operatorname{Pair}(m)}
			\prod_{\substack{\{r,s\}\in\pi\\r<s}}
			\kappa(\vec{\mathbf{x}}_r,\vec{\mathbf{x}}_s).
		\end{equation}

	For example, in the centered case the moment of four linear field factors is
	\begin{align}
		m(\vec{\mathbf{x}}_1,\vec{\mathbf{x}}_2,\vec{\mathbf{x}}_3,\vec{\mathbf{x}}_4)
		&=\kappa(\vec{\mathbf{x}}_1,\vec{\mathbf{x}}_2)\kappa(\vec{\mathbf{x}}_3,\vec{\mathbf{x}}_4)
		\nonumber\\
		&+\kappa(\vec{\mathbf{x}}_1,\vec{\mathbf{x}}_3)\kappa(\vec{\mathbf{x}}_2,\vec{\mathbf{x}}_4)
		\nonumber\\
		&+\kappa(\vec{\mathbf{x}}_1,\vec{\mathbf{x}}_4)\kappa(\vec{\mathbf{x}}_2,\vec{\mathbf{x}}_3),
	\end{align}
		which is the usual three-pairing decomposition of a four-point environmental
	correlation function, consistent with general operator forms of Wick's theorem~\cite{FerialdiDiosi2021}.

	\subsection{Linear forms}

	\begin{proposition}
		\label{prop:linear-reduction}
		Let \(\vec{\mathbf x}\in\mathbb C^{4d}\). Then, in the sense of distributions,
		\begin{align}
			&\operatorname{tr}_B\Bigl[
			(\vec{\mathbf x}^{\mathsf T}\mathfrak A)\hat W(\vec z)\Bigr]
		=i\pi^d\vec p_{\vec{\mathbf x}}^{\mathsf T}
			\partial_{\vec z}\delta^{(2d)}(\vec z).
		\end{align}
	\end{proposition}

	\begin{proof}
		For one factor, Eq.~\eqref{eq:ordered-doubled-string} gives
		\(\pi^d\mathscr D_{\vec{\mathbf x}}\delta^{(2d)}(\vec z)\).
		The coordinate part of \(\mathscr D_{\vec{\mathbf x}}\) vanishes
		because \(\vec z\,\delta^{(2d)}(\vec z)=0\), leaving the stated derivative.
	\end{proof}

	\subsection{Product of two linear forms}

	\begin{proposition}
		\label{prop:linear-product-reduction}
		Let \(\vec{\mathbf x},\vec{\mathbf y}\in\mathbb C^{4d}\).
		Then, in the sense of distributions,
		\begin{align}
			&\operatorname{tr}_B\Bigl[
			(\vec{\mathbf x}^{\mathsf T}\mathfrak A)
			(\vec{\mathbf y}^{\mathsf T}\mathfrak A)\hat W(\vec z)\Bigr]
			\nonumber\\
			&=\pi^d\Bigl[
			-(\vec p_{\vec{\mathbf y}}^{\mathsf T}\partial_{\vec z})
			(\vec p_{\vec{\mathbf x}}^{\mathsf T}\partial_{\vec z})
			+\frac12\vec q_{\vec{\mathbf y}}^{\mathsf T}J\vec p_{\vec{\mathbf x}}
			\Bigr]\delta^{(2d)}(\vec z).
		\end{align}
	\end{proposition}

	\begin{proof}
		The ordered-string identity gives
		\(\pi^d\mathscr D_{\vec{\mathbf y}}\mathscr D_{\vec{\mathbf x}}
		\delta^{(2d)}(\vec z)\).
		By Proposition~\ref{prop:linear-reduction}, the action of the
		rightmost differential operator reduces to
		\(i\vec p_{\vec{\mathbf x}}^{\mathsf T}\partial_{\vec z}
		\delta^{(2d)}(\vec z)\).
		Applying \(\mathscr D_{\vec{\mathbf y}}\) and using
		\begin{equation}
			(\vec q_{\vec{\mathbf y}}^{\mathsf T}J\vec z)
			(\vec p_{\vec{\mathbf x}}^{\mathsf T}\partial_{\vec z})
			\delta^{(2d)}(\vec z)
			=-\vec q_{\vec{\mathbf y}}^{\mathsf T}J\vec p_{\vec{\mathbf x}}
			\delta^{(2d)}(\vec z)
		\end{equation}
		yields the derivative term and the ordering contraction.
	\end{proof}

	\subsection{Quadratic forms}

	Proposition~\ref{prop:linear-reduction} removes the coordinate term
	for a single linear factor, while
	Proposition~\ref{prop:linear-product-reduction} retains the contraction
	generated by a second factor. Expanding a quadratic form into products
	of two linear forms and summing the latter result gives the following
	matrix formula.

	\begin{proposition}
		\label{prop:quadratic-reduction}
		Let \(\mathbf{X}\in\mathbb C^{4d\times4d}.\) Then
		\begin{align}
			&
			\operatorname{tr}_B
			\Bigl[
			\mathfrak A^{\mathsf T}
			\mathbf{X}
			\mathfrak A\,
			\hat W(\vec{z})
			\Bigr]
			\nonumber\\
			&=
			\pi^d
			\Bigg[
			-
			\partial_{\vec{z}}^{\mathsf T}
			P\mathbf{X}^{\mathsf T}P^{\mathsf T}
			\partial_{\vec{z}}
			+
			\frac12
			\operatorname{Tr}
			\Bigl[
			J^{\mathsf T}
			Q
			\mathbf{X}^{\mathsf T}
			P^{\mathsf T}
			\Bigr]
			\Bigg]
			\delta^{(2d)}(\vec{z}).
		\end{align}
	\end{proposition}

	\begin{proof}
		Using Eq.~\eqref{eq:rank-one-decomposition} and
		Eq.~\eqref{eq:ordered-doubled-string},
		\[
		\operatorname{tr}_B
		\Bigl[
		\mathfrak A^{\mathsf T}
		\mathbf{X}
		\mathfrak A\,
		\hat W(\vec{z})
		\Bigr]
		=
		\pi^d
		\sum_\alpha
		\mathscr D_{\vec{\mathbf{y}}_\alpha}
		\mathscr D_{\vec{\mathbf{x}}_\alpha}
		\delta^{(2d)}(\vec{z}).
		\]

		Applying Proposition~\ref{prop:linear-product-reduction} term by term,
		we collect the derivative coefficients as
		\begin{align}
			\sum_\alpha
			\vec{p}_{\vec{\mathbf{y}}_\alpha}
			\vec{p}_{\vec{\mathbf{x}}_\alpha}^{\mathsf T}
			=
			P
			\left(
			\sum_\alpha
			\vec{\mathbf{y}}_\alpha \vec{\mathbf{x}}_\alpha^{\mathsf T}
			\right)
			P^{\mathsf T}
			=
			P\mathbf{X}^{\mathsf T}P^{\mathsf T},
		\end{align}
		while the contraction term gives
		\begin{align}
			\sum_\alpha
			\vec{q}_{\vec{\mathbf{y}}_\alpha}^{\mathsf T}
			J
			\vec{p}_{\vec{\mathbf{x}}_\alpha}
			=
			\sum_\alpha
			\operatorname{Tr}
			\Bigl[
			J^{\mathsf T}
			\vec{q}_{\vec{\mathbf{y}}_\alpha}
			\vec{p}_{\vec{\mathbf{x}}_\alpha}^{\mathsf T}
			\Bigr]
			=
			\operatorname{Tr}
			\Bigl[
			J^{\mathsf T}
			Q
			\mathbf{X}^{\mathsf T}
			P^{\mathsf T}
			\Bigr].
		\end{align}

		Combining both contributions proves the statement.
	\end{proof}

	\subsection{Product of two quadratic forms}

	\begin{proposition}
		\label{prop:quadratic-product-reduction}
		Let \(\mathbf{X}_1,\mathbf{X}_2\in\mathbb C^{4d\times4d},\)
			and let \(\mathbf{R}=Q^{\mathsf T}JP.\)
		Then, in the sense of distributions,
		\begin{align}
			&
			\operatorname{tr}_B
			\Bigl[
			\left(
			\mathfrak A^{\mathsf T}\mathbf{X}_1\mathfrak A
			\right)
			\left(
			\mathfrak A^{\mathsf T}\mathbf{X}_2\mathfrak A
			\right)
			\hat W(\vec{z})
			\Bigr]
			\nonumber\\
			&=
			\pi^d
			\Bigg\{
			\left(
			\partial_{\vec{z}}^{\mathsf T}
			P\mathbf{X}_2^{\mathsf T}P^{\mathsf T}
			\partial_{\vec{z}}
			\right)
			\left(
			\partial_{\vec{z}}^{\mathsf T}
			P\mathbf{X}_1^{\mathsf T}P^{\mathsf T}
			\partial_{\vec{z}}
			\right)
			\nonumber\\
				&\quad
				-
				\frac12
			\operatorname{Tr}
			\Bigl[
			J^{\mathsf T}Q\mathbf{X}_2^{\mathsf T}P^{\mathsf T}
			\Bigr]
			\partial_{\vec{z}}^{\mathsf T}
			P\mathbf{X}_1^{\mathsf T}P^{\mathsf T}
			\partial_{\vec{z}}
			\nonumber\\
			&\quad
				-
				\frac12
			\operatorname{Tr}
			\Bigl[
			J^{\mathsf T}Q\mathbf{X}_1^{\mathsf T}P^{\mathsf T}
			\Bigr]
			\partial_{\vec{z}}^{\mathsf T}
			P\mathbf{X}_2^{\mathsf T}P^{\mathsf T}
			\partial_{\vec{z}}
			\nonumber\\
			&\quad
				-
				\frac12
			\partial_{\vec{z}}^{\mathsf T}
			P
			(\mathbf{X}_2+\mathbf{X}_2^{\mathsf T})\mathbf{R}(\mathbf{X}_1+\mathbf{X}_1^{\mathsf T})
			P^{\mathsf T}
			\partial_{\vec{z}}
			\nonumber\\
			&\quad
			+
			\frac14
			\operatorname{Tr}
			\Bigl[
			J^{\mathsf T}Q\mathbf{X}_2^{\mathsf T}P^{\mathsf T}
			\Bigr]
			\operatorname{Tr}
			\Bigl[
			J^{\mathsf T}Q\mathbf{X}_1^{\mathsf T}P^{\mathsf T}
			\Bigr]
			\nonumber\\
			&\quad
			+
			\frac14
			\operatorname{Tr}
			\Bigl[
			\mathbf{R}\mathbf{X}_1\mathbf{R}^{\mathsf T}\mathbf{X}_2^{\mathsf T}
			\Bigr]
			\nonumber\\
			&\quad+
			\frac14
			\operatorname{Tr}
			\Bigl[
			\mathbf{R}\mathbf{X}_1^{\mathsf T}\mathbf{R}^{\mathsf T}\mathbf{X}_2^{\mathsf T}
			\Bigr]
			\Bigg\}
			\delta^{(2d)}(\vec{z}).
		\end{align}
	\end{proposition}

	\begin{proof}
		Choose rank-one decompositions
		\begin{equation}
			\mathbf{X}_1
			=
			\sum_{\alpha}
			\vec{\mathbf{x}}_{\alpha}\vec{\mathbf{y}}_{\alpha}^{\mathsf T},
			\qquad
			\mathbf{X}_2
			=
			\sum_{\beta}
			\vec{\mathbf{x}}_{\beta}\vec{\mathbf{y}}_{\beta}^{\mathsf T}.
		\end{equation}
		Then
		\begin{align}
			\left(
			\mathfrak A^{\mathsf T}\mathbf{X}_1\mathfrak A
			\right)
			\left(
			\mathfrak A^{\mathsf T}\mathbf{X}_2\mathfrak A
			\right)
			=
			\sum_{\alpha,\beta}
			(\vec{\mathbf{x}}_{\alpha}^{\mathsf T}\mathfrak A)
			(\vec{\mathbf{y}}_{\alpha}^{\mathsf T}\mathfrak A)
			(\vec{\mathbf{x}}_{\beta}^{\mathsf T}\mathfrak A)
			(\vec{\mathbf{y}}_{\beta}^{\mathsf T}\mathfrak A).
		\end{align}
		Due to the reversal of the order in the differential representation,
		\begin{equation}
			(\vec{\mathbf{x}}_{\alpha}^{\mathsf T}\mathfrak A)
			(\vec{\mathbf{y}}_{\alpha}^{\mathsf T}\mathfrak A)
			(\vec{\mathbf{x}}_{\beta}^{\mathsf T}\mathfrak A)
			(\vec{\mathbf{y}}_{\beta}^{\mathsf T}\mathfrak A)
			\longmapsto
			\mathscr D_{\vec{\mathbf{y}}_\beta}
			\mathscr D_{\vec{\mathbf{x}}_\beta}
			\mathscr D_{\vec{\mathbf{y}}_\alpha}
			\mathscr D_{\vec{\mathbf{x}}_\alpha}.
		\end{equation}
			Therefore
		\begin{align}
			&
			\operatorname{tr}_B
			\Bigl[
			\left(
			\mathfrak A^{\mathsf T}\mathbf{X}_1\mathfrak A
			\right)
			\left(
			\mathfrak A^{\mathsf T}\mathbf{X}_2\mathfrak A
			\right)
			\hat W(\vec{z})
			\Bigr]
			\nonumber\\
			&=
			\pi^d
			\sum_{\alpha,\beta}
			\mathscr D_{\vec{\mathbf{y}}_\beta}
			\mathscr D_{\vec{\mathbf{x}}_\beta}
			\mathscr D_{\vec{\mathbf{y}}_\alpha}
			\mathscr D_{\vec{\mathbf{x}}_\alpha}
			\delta^{(2d)}(\vec{z}).
		\end{align}

		Applying Lemma~\ref{lem:wick-delta} to this ordered product gives
		the following contributions.
		The term without contractions gives
		\begin{equation}
			\left(
			\partial_{\vec{z}}^{\mathsf T}
			P\mathbf{X}_2^{\mathsf T}P^{\mathsf T}
			\partial_{\vec{z}}
			\right)
			\left(
			\partial_{\vec{z}}^{\mathsf T}
			P\mathbf{X}_1^{\mathsf T}P^{\mathsf T}
			\partial_{\vec{z}}
			\right)
			\delta^{(2d)}(\vec{z}).
		\end{equation}

		The contractions inside each quadratic form give
		\begin{align}
				&
				-\frac12
			\operatorname{Tr}
			\Bigl[
			J^{\mathsf T}Q\mathbf{X}_2^{\mathsf T}P^{\mathsf T}
			\Bigr]
			\partial_{\vec{z}}^{\mathsf T}
			P\mathbf{X}_1^{\mathsf T}P^{\mathsf T}
			\partial_{\vec{z}}\,\delta^{(2d)}(\vec{z})
			\nonumber\\
				&\quad
				-
				\frac12
			\operatorname{Tr}
			\Bigl[
			J^{\mathsf T}Q\mathbf{X}_1^{\mathsf T}P^{\mathsf T}
			\Bigr]
			\partial_{\vec{z}}^{\mathsf T}
			P\mathbf{X}_2^{\mathsf T}P^{\mathsf T}
			\partial_{\vec{z}}\,\delta^{(2d)}(\vec{z}).
		\end{align}

		The four possible single cross-contractions give
		\begin{align}
				&
				-\frac12
				\partial_{\vec{z}}^{\mathsf T}
			P
			(\mathbf{X}_2+\mathbf{X}_2^{\mathsf T})\mathbf{R}(\mathbf{X}_1+\mathbf{X}_1^{\mathsf T})
			P^{\mathsf T}
			\partial_{\vec{z}}\,\delta^{(2d)}(\vec{z}),
		\end{align}
		with \(\mathbf{R}\) as defined in the proposition.

		Finally, the double contractions consist of the product of the two internal
		contractions and two cross-contraction pairings. They give
		\begin{align}
			&
			\frac14
			\operatorname{Tr}
			\Bigl[
			J^{\mathsf T}Q\mathbf{X}_2^{\mathsf T}P^{\mathsf T}
			\Bigr]
			\operatorname{Tr}
			\Bigl[
			J^{\mathsf T}Q\mathbf{X}_1^{\mathsf T}P^{\mathsf T}
			\Bigr]
			\delta^{(2d)}(\vec{z})
			\nonumber\\
			&\quad
			+
			\frac14
			\operatorname{Tr}
			\Bigl[
			\mathbf{R}\mathbf{X}_1\mathbf{R}^{\mathsf T}\mathbf{X}_2^{\mathsf T}
			\Bigr]
			\delta^{(2d)}(\vec{z})
			\nonumber\\
			&\quad+
			\frac14
			\operatorname{Tr}
			\Bigl[
			\mathbf{R}\mathbf{X}_1^{\mathsf T}\mathbf{R}^{\mathsf T}\mathbf{X}_2^{\mathsf T}
			\Bigr]
			\delta^{(2d)}(\vec{z}).
		\end{align}
		Combining all contributions yields the stated formula.
	\end{proof}

	\section{\label{app:single-mode-worked-example}Single-mode calculations}

	We collect the explicit matrices used in the examples in
	Sec.~\ref{sec:single-mode-worked-example}, followed by the results for
	the other combinations of factors. All derivatives act on the initial
	characteristic function, and all two-time formulas assume
	\(0\leq t_1\leq t\).

	\subsection{Multiplication superoperators and doubled matrices}

	We use the operator vectors and generator matrices specified in
	Sec.~\ref{sec:single-mode-worked-example}.
	Evaluating the matrix exponential in Eq.~\eqref{eq:doubled-linear-propagator}
	using \(\mathbf L\) from Eq.~\eqref{eq:single-mode-liouvillian-matrix} gives
\begin{widetext}
\begin{equation}
 \mathbf{G}(t)=e^{-\mathbf{J}\mathbf{L}t}=
 \begin{pmatrix}
 e^{-(\gamma/2+i\omega_B)t}&0&0&0\\
 0&e^{(\gamma/2+i\omega_B)t}
 &\bigl(e^{-\gamma t/2}-e^{\gamma t/2}\bigr)e^{i\omega_B t}&0\\
 0&0&e^{-(\gamma/2-i\omega_B)t}&0\\
 \bigl(e^{-\gamma t/2}-e^{\gamma t/2}\bigr)e^{-i\omega_B t}
 &0&0&e^{(\gamma/2-i\omega_B)t}
 \end{pmatrix}.
 \label{eq:single-mode-doubled-propagator}
\end{equation}
\end{widetext}
	The matrices entering the differential reduction are
	\begin{equation}
		P=\begin{pmatrix}1&0&0&1\\0&1&1&0\end{pmatrix},
		\qquad
		Q=\begin{pmatrix}1&0&0&-1\\0&1&-1&0\end{pmatrix},
	\end{equation}

	Let \(\vec{\mathbf e}_j\) denote the \(j\)-th coordinate vector of
	\(\mathbb C^4\).
	For any ordered string of the doubled components, we apply
	Proposition~\ref{prop:environmental-moments} with the coefficient
	vectors \(\vec{\mathbf e}_j(t)=\mathbf G(t)^{\mathsf T}\vec{\mathbf e}_j\).
	The required derivative and contraction vectors follow directly as
	\begin{equation}
		\begin{aligned}
		\vec p_{\vec{\mathbf e}_j(t)}&=P\mathbf G(t)^{\mathsf T}\vec{\mathbf e}_j,\\
		\vec q_{\vec{\mathbf e}_j(t)}&=Q\mathbf G(t)^{\mathsf T}\vec{\mathbf e}_j.
		\end{aligned}
	\end{equation}
	Thus the general reduction applies with the explicit matrices above,
	using the same characteristic function as in the main text. Quadratic factors use
	\(\mathbf X(t)=\mathbf G(t)^{\mathsf T}\mathbf X\mathbf G(t)\).
	\subsection{Moments for linear interaction factors}
	For one factor, substituting
	\(\vec{\mathbf x}=\vec{\mathbf e}_j(t)\) into
	Proposition~\ref{prop:linear-reduction} and integrating against
	\(\chi(\vec z)\) gives
	\begin{align}
		\operatorname{tr}_B[\mathfrak A_1(t)\hat\rho_B]
		&=\operatorname{tr}_B[\mathfrak A_4(t)\hat\rho_B]
		\nonumber\\
		&=-ie^{-(\gamma/2+i\omega_B)t}
		\left.\partial_z\chi\right|_{\vec z=0},\\
		\operatorname{tr}_B[\mathfrak A_2(t)\hat\rho_B]
		&=\operatorname{tr}_B[\mathfrak A_3(t)\hat\rho_B]
		\nonumber\\
		&=-ie^{-(\gamma/2-i\omega_B)t}
		\left.\partial_{z^*}\chi\right|_{\vec z=0}.
	\end{align}
	For two factors, Proposition~\ref{prop:linear-product-reduction} with
	\(\vec{\mathbf x}=\vec{\mathbf e}_i(t)\) and
	\(\vec{\mathbf y}=\vec{\mathbf e}_j(t_1)\) gives
	Eq.~\eqref{eq:single-mode-linear-two-time} for \(i=1,j=2\).
	The remaining combinations are
	\begin{subequations}
	\label{eq:single-mode-linear-other-moments}
	\begin{align}
		&\operatorname{tr}_B[\mathfrak A_2(t)\mathfrak A_1(t_1)\hat\rho_B]
		\nonumber\\
		&\quad=e^{-(\gamma/2-i\omega_B)(t-t_1)}e^{-\gamma t_1}
		\nonumber\\
		&\qquad{}\times
		\left[-\left.\partial_z\partial_{z^*}\chi\right|_{\vec z=\vec{0}}-\frac12\right],\\
		&\operatorname{tr}_B[\mathfrak A_1(t)\mathfrak A_3(t_1)\hat\rho_B]
		\nonumber\\
		&\quad=e^{-(\gamma/2+i\omega_B)(t-t_1)}e^{-\gamma t_1}
		\nonumber\\
		&\qquad{}\times
		\left[-\left.\partial_z\partial_{z^*}\chi\right|_{\vec z=\vec{0}}-\frac12\right],\\
		&\operatorname{tr}_B[\mathfrak A_2(t)\mathfrak A_4(t_1)\hat\rho_B]
		\nonumber\\
		&\quad=e^{-(\gamma/2-i\omega_B)(t-t_1)}
		\nonumber\\
		&\qquad{}\times
		\left[-e^{-\gamma t_1}\left.\partial_z\partial_{z^*}\chi\right|_{\vec z=\vec{0}}
		+1-\frac12e^{-\gamma t_1}\right],\\
		&\operatorname{tr}_B[\mathfrak A_1(t)\mathfrak A_1(t_1)\hat\rho_B]
		\nonumber\\
		&\quad=\operatorname{tr}_B[\mathfrak A_1(t)\mathfrak A_4(t_1)\hat\rho_B]
		\nonumber\\
		&\quad=-e^{-(\gamma/2+i\omega_B)(t+t_1)}
		\left.\partial_z^2\chi\right|_{\vec z=\vec{0}},\\
		&\operatorname{tr}_B[\mathfrak A_2(t)\mathfrak A_2(t_1)\hat\rho_B]
		\nonumber\\
		&\quad=\operatorname{tr}_B[\mathfrak A_2(t)\mathfrak A_3(t_1)\hat\rho_B]
		\nonumber\\
		&\quad=-e^{-(\gamma/2-i\omega_B)(t+t_1)}
		\left.\partial_{z^*}^{\,2}\chi\right|_{\vec z=\vec{0}}.
	\end{align}
	\end{subequations}
	Inside the trace, the latest factor \(\mathfrak A_1(t)\) can be
	replaced by \(\mathfrak A_4(t)\), and \(\mathfrak A_2(t)\) by
	\(\mathfrak A_3(t)\). This follows from trace preservation and
	cyclicity, and supplies all remaining left--right combinations.

	\subsection{Moments for quadratic interaction factors}

	We now apply the matrix reduction formulas to the quadratic factors.
	The matrices
	\begin{equation}
		\mathbf X_1=\vec{\mathbf e}_1\vec{\mathbf e}_1^{\mathsf T},
		\qquad
		\mathbf X_2=\vec{\mathbf e}_2\vec{\mathbf e}_2^{\mathsf T}
	\end{equation}
	select the squares of the first and second doubled components, with
	\begin{equation}
		\mathbf X_j(t)=\mathbf G(t)^{\mathsf T}\mathbf X_j\mathbf G(t),
		\qquad j=1,2.
	\end{equation}
	For a single factor, Proposition~\ref{prop:quadratic-reduction} requires
	the derivative matrix and the internal contraction. For \(\mathbf X_1(t)\),
	these are
	\begin{align}
		P\mathbf X_1(t)^{\mathsf T}P^{\mathsf T}
		&=
		\begin{pmatrix}e^{-(\gamma+2i\omega_B)t}&0\\0&0\end{pmatrix},\\
		\operatorname{Tr}\!\left[J^{\mathsf T}Q\mathbf X_1(t)^{\mathsf T}P^{\mathsf T}\right]
		&=0.
	\end{align}
As a result
	\begin{align}
		\operatorname{tr}_B[\mathfrak A_1(t)^2\hat\rho_B]
=-e^{-(\gamma+2i\omega_B)t}
		\left.\partial_z^2\chi(\vec z)\right|_{\vec z=\vec{0}}.
	\end{align}
	Applying the same proposition to the remaining squared components gives
	the single-factor results
	\begin{subequations}
	\label{eq:single-mode-quadratic-one-time}
	\begin{align}
		\operatorname{tr}_B[\mathfrak A_1(t)^2\hat\rho_B]
		&=\operatorname{tr}_B[\mathfrak A_4(t)^2\hat\rho_B]
		\nonumber\\
		&=-e^{-(\gamma+2i\omega_B)t}\left.\partial_z^2\chi\right|_{\vec z=\vec{0}},\\
		\operatorname{tr}_B[\mathfrak A_2(t)^2\hat\rho_B]
		&=\operatorname{tr}_B[\mathfrak A_3(t)^2\hat\rho_B]
		\nonumber\\
		&=-e^{-(\gamma-2i\omega_B)t}\left.\partial_{z^*}^{\,2}\chi\right|_{\vec z=\vec{0}}.
	\end{align}
	\end{subequations}

	For the ordered product \(\mathfrak A_1(t)^2\mathfrak A_2(t_1)^2\),
	we apply Proposition~\ref{prop:quadratic-product-reduction} with
	\(\mathbf X_1(t)\) as its first matrix and \(\mathbf X_2(t_1)\) as
	its second matrix. In addition to the matrices above, we need
	\begin{align}
		P\mathbf X_2(t_1)^{\mathsf T}P^{\mathsf T}
		&=
		\begin{pmatrix}0&0\\0&e^{-(\gamma-2i\omega_B)t_1}\end{pmatrix},\\
		\operatorname{Tr}\!\left[J^{\mathsf T}Q\mathbf X_2(t_1)^{\mathsf T}P^{\mathsf T}\right]
		&=0.
	\end{align}
	Both internal contractions vanish. Since the two propagated matrices
	are symmetric, the remaining terms of the proposition simplify to
	the following three contributions. The product of the derivative
	parts gives
	\begin{align}
		&\left(\partial_{\vec z}^{\mathsf T}P\mathbf X_2(t_1)^{\mathsf T}P^{\mathsf T}\partial_{\vec z}\right)
		\left(\partial_{\vec z}^{\mathsf T}P\mathbf X_1(t)^{\mathsf T}P^{\mathsf T}\partial_{\vec z}\right)
		\nonumber\\
		&=e^{-(\gamma+2i\omega_B)(t-t_1)}e^{-2\gamma t_1}
		\partial_z^2\partial_{z^*}^{\,2}.
	\end{align}
	The single cross-contractions give
	\begin{align}
		&-2\partial_{\vec z}^{\mathsf T}P\mathbf X_2(t_1)\mathbf R\mathbf X_1(t)P^{\mathsf T}\partial_{\vec z}
		\nonumber\\
		&=e^{-(\gamma+2i\omega_B)(t-t_1)}
		2e^{-\gamma t_1}(e^{-\gamma t_1}-2)\partial_z\partial_{z^*},
	\end{align}
	and the two double cross-contractions give
	\begin{align}
		&\frac12\operatorname{Tr}\!\left[\mathbf R\mathbf X_1(t)\mathbf R^{\mathsf T}\mathbf X_2(t_1)\right]
		\nonumber\\
		&=\frac12e^{-(\gamma+2i\omega_B)(t-t_1)}(e^{-\gamma t_1}-2)^2,
	\end{align}
	where \(\mathbf R=Q^{\mathsf T}JP\) is the matrix defined in the proposition.
	Consequently,
	\begin{align}
		&\operatorname{tr}_B[\mathfrak A_1(t)^2\mathfrak A_2(t_1)^2\hat W(\vec z)]
		\nonumber\\
		&=\pi e^{-(\gamma+2i\omega_B)(t-t_1)}
		\Bigl[e^{-2\gamma t_1}\partial_z^2\partial_{z^*}^{\,2}
		\nonumber\\
		&\qquad+2e^{-\gamma t_1}(e^{-\gamma t_1}-2)\partial_z\partial_{z^*}
		\nonumber\\
		&\qquad+\frac12(e^{-\gamma t_1}-2)^2\Bigr]\delta^{(2)}(\vec z).
	\end{align}
	Substitution of the inverse Weyl transform and integration by parts
	transfer the derivatives to \(\chi(\vec z)\). Using \(\chi(0)=1\) in
	the constant term gives Eq.~\eqref{eq:single-mode-quadratic-two-time}.
	At \(t=t_1=0\),
	\begin{align}
		&\operatorname{tr}_B[\hat a^2(\hat a^\dagger)^2\hat\rho_B] = 
		\left.\partial_z^2\partial_{z^*}^{\,2}\chi\right|_{\vec z=0}
		-2\left.\partial_z\partial_{z^*}\chi\right|_{\vec z=\vec{0}}+\frac12.
		\label{eq:single-mode-static-moment}
	\end{align}
	For the vacuum characteristic function \(\chi=e^{-zz^*/2}\), this
	gives \(2\). More generally, substituting the same characteristic
	function into the two-time result gives
	\(2e^{-(\gamma+2i\omega_B)(t-t_1)}\).

	For products of two quadratic factors, the other combinations are
	\begin{subequations}
	\label{eq:single-mode-quadratic-other-moments}
	\begin{align}
		&\operatorname{tr}_B[\mathfrak A_2(t)^2\mathfrak A_1(t_1)^2\hat\rho_B]
		\nonumber\\
		&=e^{-(\gamma-2i\omega_B)(t-t_1)}e^{-2\gamma t_1}
		\nonumber\\
		&\quad\times\left[\left.\partial_z^2\partial_{z^*}^{\,2}\chi\right|_{\vec z=\vec{0}}
		+2\left.\partial_z\partial_{z^*}\chi\right|_{\vec z=0}+\frac12\right],
		\displaybreak[0]\\
		&\operatorname{tr}_B[\mathfrak A_1(t)^2\mathfrak A_3(t_1)^2\hat\rho_B]
		\nonumber\\
		&=e^{-(\gamma+2i\omega_B)(t-t_1)}e^{-2\gamma t_1}
		\nonumber\\
		&\quad\times\left[\left.\partial_z^2\partial_{z^*}^{\,2}\chi\right|_{\vec z=\vec{0}}
		+2\left.\partial_z\partial_{z^*}\chi\right|_{\vec z=\vec{0}}+\frac12\right],
		\displaybreak[0]\\
		&\operatorname{tr}_B[\mathfrak A_2(t)^2\mathfrak A_4(t_1)^2\hat\rho_B]
		\nonumber\\
		&=e^{-(\gamma-2i\omega_B)(t-t_1)}
		\nonumber\\
		&\quad\times\Bigl[e^{-2\gamma t_1}
		\left.\partial_z^2\partial_{z^*}^{\,2}\chi\right|_{\vec z=\vec{0}}
		\nonumber\\
		&\qquad+2e^{-\gamma t_1}(e^{-\gamma t_1}-2)
		\left.\partial_z\partial_{z^*}\chi\right|_{\vec z=\vec{0}}
		\nonumber\\
		&\qquad+\frac12(e^{-\gamma t_1}-2)^2\Bigr],
		\displaybreak[0]\\
		&\operatorname{tr}_B[\mathfrak A_1(t)^2\mathfrak A_1(t_1)^2\hat\rho_B]
=\operatorname{tr}_B[\mathfrak A_1(t)^2\mathfrak A_4(t_1)^2\hat\rho_B]
		\nonumber\\
		&=e^{-(\gamma+2i\omega_B)(t+t_1)}
		\left.\partial_z^4\chi\right|_{\vec z=\vec{0}},
		\displaybreak[0]\\
		&\operatorname{tr}_B[\mathfrak A_2(t)^2\mathfrak A_2(t_1)^2\hat\rho_B]
=\operatorname{tr}_B[\mathfrak A_2(t)^2\mathfrak A_3(t_1)^2\hat\rho_B]
		\nonumber\\
		&=e^{-(\gamma-2i\omega_B)(t+t_1)}
		\left.\partial_{z^*}^{\,4}\chi\right|_{\vec z=\vec{0}}.
	\end{align}
	\end{subequations}
	As for the linear factors, replacing the latest factor
	\(\mathfrak A_1(t)^2\) by \(\mathfrak A_4(t)^2\), or
	\(\mathfrak A_2(t)^2\) by \(\mathfrak A_3(t)^2\), leaves the trace
	unchanged. These formulas therefore cover both left and right
	multiplication channels of the two-photon interaction.

\section{\label{app:two-photon-rabi-details}Environmental correlations for the two-photon Rabi model}
\label{app:rabi-environmental-moments}

This appendix evaluates the environmental correlation functions for the
model in Sec.~\ref{sec:two-photon-rabi-model}.

The free environmental generator and doubled propagator are those of
Sec.~\ref{sec:single-mode-worked-example}. In particular,
\(\mathbf V_j(t)=\mathbf G(t)^{\mathsf T}\mathbf V_j\mathbf G(t)\).
The diagonal channel matrices select
\begin{align}
 \mathcal V_1^B(t)&=\mathfrak A_1(t)^2+\mathfrak A_2(t)^2,\\
 \mathcal V_2^B(t)&=\mathfrak A_3(t)^2+\mathfrak A_4(t)^2.
\end{align}
Adding the single-factor results in
Eq.~\eqref{eq:single-mode-quadratic-one-time} gives
Eq.~\eqref{eq:first-order-characteristic-average} for both channels.
The identity
\(-\left.\partial_z^2\chi\right|_{\vec z=0}=\mu_1^2+\Sigma_{11}\),
which follows from the definition of the field mean and covariance,
then gives Eq.~\eqref{eq:gaussian-first-order-average}.

At second order, trace preservation and cyclicity make the result
independent of the latest channel \(l\). The two choices of the earlier
channel give
\begin{align}
 m_{l1}(t,t_1)
 &=\operatorname{tr}_B\Bigl[
 (\mathfrak A_1(t)^2+\mathfrak A_2(t)^2)
 \nonumber\\
 &\qquad\times(\mathfrak A_1(t_1)^2+\mathfrak A_2(t_1)^2)\hat\rho_B\Bigr],\\
 m_{l2}(t,t_1)
 &=\operatorname{tr}_B\Bigl[
 (\mathfrak A_1(t)^2+\mathfrak A_2(t)^2)
 \nonumber\\
 &\qquad\times(\mathfrak A_3(t_1)^2+\mathfrak A_4(t_1)^2)\hat\rho_B\Bigr].
\end{align}
Each expression is a sum of four moments already evaluated by the
matrix and differential procedure. Substituting
Eqs.~\eqref{eq:single-mode-quadratic-two-time}
and~\eqref{eq:single-mode-quadratic-other-moments} and combining the
opposite oscillatory phases gives, for \(0\leq t_1\leq t\) and either
\(l=1,2\),
\begin{widetext}
\begin{align}
	m_{lk}(t,t_1)
	&=e^{-\gamma(t+t_1)}
	\Bigl[
	e^{-2i\omega_B(t+t_1)}
	\left.\partial_z^4\chi\right|_{\vec{z}=\vec{0}}
	+
	e^{2i\omega_B(t+t_1)}
	\left.\partial_{z^*}^{\,4}\chi\right|_{\vec{z}=\vec{0}}
	\Bigr]
	\nonumber\\
	&\quad+e^{-\gamma(t-t_1)}
	\Bigl\{
	2e^{-2\gamma t_1}\cos\!\bigl[2\omega_B(t-t_1)\bigr]\,
	\left.\partial_z^2\partial_{z^*}^{\,2}\chi\right|_{\vec{z}=\vec{0}}
	\nonumber\\
	&\qquad+4e^{-\gamma t_1}\bigl[(e^{-\gamma t_1}-1)\cos\!\bigl[2\omega_B(t-t_1)\bigr]
	+i(-1)^{k-1}\sin\!\bigl[2\omega_B(t-t_1)\bigr]\bigr]
	\left.\partial_z\partial_{z^*}\chi\right|_{\vec{z}=\vec{0}}
	\nonumber\\
	&\qquad+(e^{-2\gamma t_1}-2e^{-\gamma t_1}+2)\cos\!\bigl[2\omega_B(t-t_1)\bigr]
	+2i(-1)^{k-1}(e^{-\gamma t_1}-1)\sin\!\bigl[2\omega_B(t-t_1)\bigr]
	\Bigr\}.
	\label{eq:two-photon-second-order-moments}
\end{align}
\end{widetext}
Here \(\chi\) is the initial characteristic function, with
\(\chi(0)=1\) used in the constant term. The different placement of
the two mixed quadratic moments accounts for the sign \((-1)^{k-1}\),
which distinguishes left and right multiplication at the earlier time.

\section{\label{app:initial-state-derivatives}Initial-state moments and coupling diagnostic}

This appendix gives characteristic-function derivatives at arbitrary order
for Fock and Gaussian states. Their fourth-order derivatives are then
used in the Fock--thermal comparison and in the weak-coupling diagnostic
for the numerical simulations.

\subsection{Fock states}

For \(p,q\in\mathbb N_0\), the derivatives of the Fock-state characteristic
function given in Sec.~\ref{sec:matched-fock-thermal-states} satisfy
\begin{equation}
\begin{aligned}
&
	(-i)^{p+q}
	\left.
	\frac{\partial^{p+q}\chi_{n_B}}
	{\partial z^p\partial(z^*)^q}
	\right|_{\vec{z}=\vec{0}}
	\\
	&=
	\delta_{p,q}
	p!
	\sum_{j=0}^{\min(p,n_B)}
	\binom{p}{j}
	\binom{n_B}{j}
	2^{-(p-j)}.
\end{aligned}
\end{equation}

\subsection{Gaussian states}
For a one-mode Gaussian state with characteristic function
\eqref{eq:gaussian-characteristic-function}, the derivatives can be
computed recursively from its mean \(\vec\mu\) and covariance \(\Sigma\).
For \(p,q\in\mathbb N_0\), define the normalized derivatives
\begin{equation}
	\chi_{\mathrm G}^{(p,q)}
	:=
	(-i)^{p+q}
	\left.
	\frac{\partial^{p+q}\chi_{\mathrm G}}
	{\partial z^p\partial(z^*)^q}
	\right|_{\vec{z}=\vec{0}},
	\qquad
	\chi_{\mathrm G}^{(0,0)}=1.
\end{equation}
They are generated at arbitrary order by
\begin{align}
	\chi_{\mathrm G}^{(p+1,q)}
	&=
	\mu_1\chi_{\mathrm G}^{(p,q)}
	+
	p\Sigma_{11}\chi_{\mathrm G}^{(p-1,q)}
	+
	q\Sigma_{12}\chi_{\mathrm G}^{(p,q-1)},
	\\
	\chi_{\mathrm G}^{(p,q+1)}
	&=
	\mu_2\chi_{\mathrm G}^{(p,q)}
	+
	p\Sigma_{12}\chi_{\mathrm G}^{(p-1,q)}
	+
	q\Sigma_{22}\chi_{\mathrm G}^{(p,q-1)},
\end{align}
where a quantity with a negative index is understood to be zero.


\subsection{\label{app:rabi-coupling-diagnostic}Numerical weak-coupling diagnostic for the two-photon Rabi model}
\label{app:matched-fourth-derivatives}

The fourth derivatives used in Sec.~\ref{sec:matched-fock-thermal-states}
follow by setting \(p=q=2\) in the formulas above. For the thermal state
with \(n_{\mathrm{th}}=n_B\), the Gaussian recurrence is evaluated at
\(\vec\mu=\vec{0}\), \(\Sigma_{11}=\Sigma_{22}=0\), and
\(\Sigma_{12}=n_B+1/2\). This gives
\begin{align}
 \left.\partial_z^2\partial_{z^*}^{\,2}\chi_{n_B}\right|_{\vec{z}=0}
 &=n_B(n_B+1)+\frac12,\\
 \left.\partial_z^2\partial_{z^*}^{\,2}\chi_{\mathrm{th}}\right|_{\vec{z}=\vec{0}}
 &=2n_B(n_B+1)+\frac12.
\end{align}
Their difference gives Eq.~\eqref{eq:fock-thermal-fourth-difference}.
The derivatives $\partial_z^4\chi|_{\vec{z}=\vec{0}}$ and
$\partial_{z^*}^{\,4}\chi|_{\vec{z}=\vec{0}}$ vanish for both states, while
$\partial_z\partial_{z^*}\chi|_{\vec{z}=0}=-(n_B+1/2)$ is common to them.

These derivatives also determine the correlation used to estimate the
coupling strength for the numerical simulations. Setting \(t_1=0\) in
Eq.~\eqref{eq:two-photon-second-order-moments} and substituting the
Fock-state derivatives above gives
\begin{align}
	m_{11}^{\mathrm F}(\tau,0)
	&=
	e^{-\gamma\tau}
	\left[
	(n_B+1)(n_B+2)e^{-2i\omega_B\tau}\right.
	\nonumber\\
	&\left.{}+
	n_B(n_B-1)e^{2i\omega_B\tau}
	\right].
\end{align}
Consequently,
\begin{equation}
	\left|m_{11}^{\mathrm F}(\tau,0)\right|
	\leq
	2\left(n_B^2+n_B+1\right)e^{-\gamma\tau}.
\end{equation}
On the observation interval \([0,T]\), define the envelope memory factor
\begin{equation}
	f_\gamma(T)
	=
	\int_0^T e^{-\gamma\tau}\,d\tau
	=
	\begin{cases}
		\dfrac{1-e^{-\gamma T}}{\gamma},
		& \gamma>0,
		\\[1ex]
		T,
		& \gamma=0.
	\end{cases}
\end{equation}
The numerical weak-coupling diagnostic is
\begin{equation}
	\eta_{\mathrm{num}}(T)
	=
	|\lambda g|
	\sqrt{2\left(n_B^2+n_B+1\right)}
	f_\gamma(T).
\end{equation}
The conservative weak-coupling regime corresponds to
\begin{equation}
	\eta_{\mathrm{num}}(T)\ll1
\end{equation}
on the full simulation interval. The envelope bound neglects oscillatory cancellations. We use the resulting criterion for parameter selection as described in Sec.~\ref{sec:numerical-comparison}.

\end{document}